\documentclass[lettersize,journal]{IEEEtran}
\usepackage{amsmath,amsfonts}
\usepackage{algorithmic}
\usepackage{algorithm}
\usepackage{array}
\usepackage{textcomp}
\usepackage{stfloats}
\usepackage{url}
\usepackage{verbatim}
\usepackage{graphicx}
\usepackage{cite}
\usepackage{tikz,pgfplots}
\usetikzlibrary{positioning, shapes, arrows, calc}
\usepackage[font=footnotesize]{caption}
\usepackage[font=footnotesize]{subcaption}
\usepgfplotslibrary{fillbetween}
\usepackage{amsthm}
\usepackage{circuitikz}
\usepackage{multirow} 

    \definecolor{my_pink}{rgb}{1,0.07,0.65}
\definecolor{my_g}{rgb}{0.09,0.71,0.14}

\usepackage[none]{hyphenat} 

\newcommand{\triangleq}{\stackrel{\triangle}{=}}

\newtheoremstyle{italic-title} 
{} 
{} 
{\normalfont} 
{} 
{\itshape} 
{:} 
{ } 
{} 

\newtheorem{theorem}{Theorem}
\newtheorem{corollary}{Corollary}

\begin{document}

\title{On Error Rate Approximations for FSO Systems with Weak Turbulence and Pointing Errors}

\author{Carmen Álvarez Roa, Yunus Can Gültekin,~\IEEEmembership{Member,~IEEE,} Kaiquan Wu, Cornelis Willem Korevaar, \\and Alex Alvarado,~\IEEEmembership{Senior Member,~IEEE}
\thanks{
This work is part of the project BIT-FREE with file number 20348 of the research programme Open Technology Programme which is (partly) financed by the Dutch Research Council (NWO).
This article was presented in part at the ICTON, Bari, Italy, June 2024. \emph{(Corresponding author: Carmen Álvarez Roa.)}}
\thanks{The authors are with the Information and Communication Theory Lab, Signal Processing Systems Group, Department of Electrical Engineering, Eindhoven University of Technology, Eindhoven 5600 MB, The Netherlands (e-mails: \{c.alvarez.roa, y.c.g.gultekin, k.wu, c.w.korevaar, a.alvarado\}@tue.nl).}
}




\maketitle

\begin{abstract}
Atmospheric attenuation, atmospheric turbulence, geometric spread, and pointing errors, degrade the performance of free-space optical transmission. 
In the weak turbulence regime, the probability density function describing the distribution of the channel fading coefficient that models these four effects is known in the literature. This function is an integral equation, which makes it difficult to find simple analytical expressions of important performance metrics such as the bit error rate (BER) and symbol error rate (SER). 
In this paper, we present simple and accurate approximations of the average BER and SER for pulse-amplitude modulation (PAM) in the weak turbulence regime for an intensity modulation and direct detection system. 
Our numerical results show that the proposed expressions exhibit excellent accuracy when compared against Monte Carlo simulations. 
To demonstrate the usefulness of the developed approximations, we perform two asymptotic analyses. First, we investigate the additional transmit power required to maintain the same SER when the spectral efficiency increases by $1$ bit/symbol. Second, we study the asymptotic behavior of our SER approximation for dense PAM constellations and high transmit power. 

\end{abstract}


 \section{Introduction}

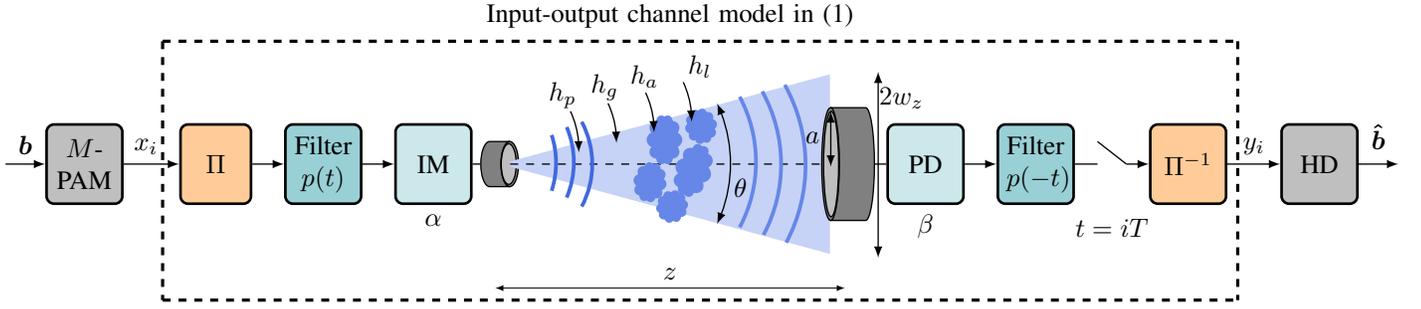
\begin{figure*}[h]
    \centering
   \usetikzlibrary {shapes.geometric}

\usetikzlibrary{circuits.ee.IEC}

\definecolor{ashgrey}{rgb}{0.75, 0.75, 0.75}
\definecolor{antiquebrass}{rgb}{0.98, 0.81, 0.69}
\definecolor{brilliantlavender}{rgb}{0.76, 0.53, 0.85}
\definecolor{RoyalBlue}{rgb}{0.25, 0.41, 0.88}
\definecolor{JungleGreen}{rgb}{0.0, 0.54, 0.6}

\tikzstyle{block} = [draw, line width = 1pt, fill=black!20, rectangle, minimum height=30pt, rounded corners=0.1cm, text width=2.2em,align=center]


\tikzstyle{block11} = [draw, line width = 1pt, fill=black!20, rectangle, minimum height=30pt, minimum width=15pt, rounded corners=0.1cm,text width=2em,align=center]
    
\tikzstyle{block2} = [draw, line width = 1pt, fill=black!20, rectangle, minimum height=30pt, minimum width=30pt, rounded corners=0.1cm, text width=5em,align=center]

\tikzstyle{Cir} = [draw, circle,  minimum size=2.15em]

\begin{tikzpicture}[auto, node distance=1 cm,>=to,line width=0.5pt]

    \node [coordinate] (input) {};
    \node [block, right = 1.5em of input, fill = gray!50] (OOK) {$M$-\\PAM}; 
    \node [block11,right = 2.1em of OOK, fill = orange!40] (INTER) {$\Pi$};
    \node [block,right = 1.2em of INTER, fill = JungleGreen!40 ] (FILTER) {Filter $p(t)$};
    \node [block,right = 1.2em of FILTER, fill = JungleGreen!20, label={[black]below:$\alpha$}] (EO) {IM};
    
    \node [block, right = 5.5 cm of EO, fill = JungleGreen!20, label={[black]below:$\beta$}] (OE) {PD}; 
    \node [block, right = 1.2 em of OE, fill = JungleGreen!40] (FILTER2) {Filter $p(-t)$}; 
    \node [block, right = 2.8em of  FILTER2, fill = orange!40 ] (INTER2) {$\Pi^{-1}$};
    \node [block, right = 2em of INTER2, fill = gray!50 ] (HD) {HD};
    \node [right = 1.5em of  HD] (output) {};


    \draw [draw,-latex] (input) -- node[midway,above]{$\boldsymbol{b}$} (OOK);
    \draw [draw,-latex] (OOK) -- node[midway,above, xshift=-0.2em]{$x_i$} (INTER); 
    \draw [draw,-latex] (INTER) -- node[midway,above]{}(FILTER); 
    \draw [draw,-latex] (FILTER) -- node[midway,above]{}(EO); 
    

    \draw [draw,-latex] (OE) -- node[midway,above]{}(FILTER2);
    \draw [draw,-] (FILTER2) -- node[midway,above]{}($(FILTER2.east)+(0.3,0)$);
    \draw [draw,-latex] ($(INTER2.west)+(-0.3,0)$) -- (INTER2);
    \draw [draw,-] ($(FILTER2.east)+(0.3,0.3)$) -- node[midway,below, yshift=-0.75cm]{$t=iT$} ($(INTER2.west)+(-0.3,0)$);
    \draw [draw,-latex] (INTER2) -- node[midway,above]{$y_i$} (HD);
    \draw [draw,-latex] (HD) --  node[midway,above,yshift=0.2em,text width=5em, align = center]{$\boldsymbol{\hat{b}}$}(output);



    \draw [draw,black,line width=0.1pt,latex-latex] ($(EO.east)+(0.3,-1.65)$) -- node[midway,above=9.5em,black,text width=15.5em, align = center]{Input-output channel model in \eqref{ChannelEq}} node[midway,above=0em,black,text width=5.5em, align = center]{$z$} ($(OE.west)+(-0.55,-1.65)$);

    \draw[dashed,very thick] ($(EO.south)+(-3.6,6.2em)$) -- ($(EO.south)+(-3.6,-3.6em)$);
    \draw[dashed,very thick] ($(EO.south)+(-3.6,6.2em)$) -- ($(OE.south)+(4.15,6.2em)$);
    \draw[dashed,very thick] ($(OE.south)+(4.15,6.2em)$) -- ($(OE.south)+(4.15,-3.6em)$);
    \draw[dashed,very thick] ($(EO.south)+(-3.6,-3.6em)$) -- ($(OE.south)+(4.15,-3.6em)$);



\node [cylinder, right = 0.1 cm of EO, shape border rotate=0, draw, minimum width=0.6 cm, minimum height=0.1cm, fill=gray, opacity=1] (ApertTX)  {};

\draw [solid, line width=0.5pt, fill=ashgrey, opacity=1] 
    ($(ApertTX.center)+(0.23cm,0)$) ellipse (0.1cm and 0.27cm) node[at=($(ApertTX.south)$), yshift=-1cm, xshift=0.2cm, scale=1]{};
        
       \fill[RoyalBlue!30] 
    ($(EO.east)+(0.5,0.1em)$) -- 
    ($(OE.west)+(-0.75,3.4em)$) -- 
    ($(OE.west)+(-0.75,-3.4em)$) -- 
    ($(EO.east)+(0.5,-0.1em)$) -- cycle;
      
    \draw[dashed,line width=0.5pt] ($(EO.east)+(0.5,0em)$) -- ($(OE.west)+(-0.5,0em)$);
    
 \node [cylinder,left = 0.15 cm of OE, shape border rotate=180, draw, minimum width=1.5cm, minimum height=0.7cm, fill=gray, opacity=1] (ApertRX)  {};

\draw [solid, line width=0.5pt, fill=ashgrey, opacity=1] 
    ($(ApertRX.center)+(-0.32cm,0)$) ellipse (0.1cm and 0.7cm) node[at=($(ApertRX.south)$), yshift=-0.6cm, xshift=0.2cm, scale=1]{};

\draw[solid,color=RoyalBlue, line width=0.5mm]  ($(ApertTX.center)+(1,1.4em)$) arc(23:-23:1.2) node[midway,right]{};
\draw[solid,color=RoyalBlue, line width=0.5mm]  ($(ApertTX.center)+(0.8,1.1em)$) arc(18:-18:1.2) node[midway,right]{};
\draw[solid,color=RoyalBlue, line width=0.5mm]  ($(ApertTX.center)+(1.2,1.6em)$) arc(28:-28:1.2) node[midway,right]{};
\draw[solid,color=RoyalBlue!80, line width=0.5mm]  ($(ApertTX.center)+(3.3,2.55em)$) arc(28:-28:1.85) node[midway,right]{};
\draw[solid,color=RoyalBlue!80, line width=0.5mm]  ($(ApertTX.center)+(3.6,2.75em)$) arc(29:-28:2) node[midway,right]{};
\draw[solid,color=RoyalBlue!80, line width=0.5mm]  ($(ApertTX.center)+(3.9,3em)$) arc(29:-28:2.2) node[midway,right]{};

\draw[latex-latex]  ($(EO.east)+(3.25,2.3em)$) arc(25:-25:1.9) node[midway,right, xshift=-0.5ex,yshift=-2ex]{$\theta$};

\node [cloud, draw=RoyalBlue!80,left = 3.1 cm of OE, fill=RoyalBlue!80, aspect=1.4, rotate=80]  {};
\node [cloud, draw=RoyalBlue!80,left = 2.45 cm of OE,yshift=1.5ex, fill=RoyalBlue!80, aspect=2, rotate=70]  {};
\node [cloud, draw=RoyalBlue!80,left = 2.75 cm of OE,yshift=-2.5ex, fill=RoyalBlue!80, aspect=1.1, rotate=60]  {};
\node [cloud, draw=RoyalBlue!80,left = 2.95 cm of OE,yshift=4ex, fill=RoyalBlue!80, aspect=1.8, rotate=90]  {};
\node [cloud, draw=RoyalBlue!80,left = 2.5 cm of OE,yshift=4.5ex, fill=RoyalBlue!80, aspect=1.2, rotate=100]  {};

 \draw[solid,line width=0.5pt, latex-latex] ($(OE.south)+(-1.26,3.55em)$) -- ($(OE.south)+(-1.26,1.45em)$)node[at=($(OE.west)$),yshift=1em, xshift=-0.98cm,scale=1]{$a$};

 \draw[solid,line width=0.5pt, latex-latex] ($(OE.south)+(-0.63,5.05em)$) -- ($(OE.south)+(-0.63,-2em)$)node[at=($(OE.west)$),yshift=2.5em, xshift=0.2cm,scale=1]{$2w_z$};;

\draw[draw, -latex]  ($(EO.east)+(2.3,2.7em)$) arc(25:-1:1.2) node[at=($(EO.east)$),xshift=6.5em, yshift=3.3em,scale=1]{$h_{a}$};

\draw[draw, -latex]  ($(EO.east)+(1.8,2.3em)$) arc(25:-1:1.2) node[at=($(EO.east)$),xshift=5em, yshift=2.9em,scale=1]{$h_{g}$};

\draw[draw, -latex]  ($(EO.east)+(1.3,2em)$) arc(25:-2:1.2) node[at=($(EO.east)$),xshift=3.4em, yshift=2.6em,scale=1]
{$h_{p}$};

\draw[draw, -latex]  ($(EO.east)+(2.8,3.1em)$) arc(25:-1:1.2) node[at=($(EO.east)$),xshift=8.6em, yshift=3.7em,scale=1]{$h_{l}$};

\draw [draw,-] (EO) --  ($(EO.east)+(0.1,0)$); 
\draw [draw,-] (OE) --  ($(OE.west)+(-0.15,0)$); 

\end{tikzpicture}  
    \caption{Block diagram of the IM/DD FSO communication system with $M$-PAM modulation under consideration. The channel effects consist of atmospheric turbulence ($h_{a}$), geometric spread ($h_{g}$), pointing errors ($h_{p}$), and atmospheric losses ($h_{l}$).}
    \label{fig:system}
\end{figure*}

Free-space optical (FSO) communication systems are currently considered a promising complementary technology to radio frequency (RF) links~\cite{Ghassemlooy15}. These systems offer a mobile broadband solution to meet the demands of massive connectivity and large data traffic in the era of beyond 5G and 6G technologies.
The growing interest in FSO communications is mainly due to the vast available bandwidth which is much larger than that of RF-based wireless systems, promising ultra-high data rates on the order of several Gbps~\cite{5G}. Furthermore, FSO systems offer high security due to their narrow beam width, along with easy deployment, low cost, and low power consumption \cite[Sec. 3.1]{AdvantFSO}. Additionally, the spectrum above 300 GHz does not require licensing, allowing operators to access optical bands without paying for exclusive rights~\cite{Survey}. 

The main disadvantage of FSO is the unreliability of the communication channel. This unreliability is caused by various factors, including atmospheric attenuation, atmospheric turbulence, geometric spread, and pointing errors~\cite{Farid2007}. The channel model studied in this work considers these four effects as the main factors that independently create multiplicative channel gain coefficients, ultimately limiting the performance of the underlying communication system.

Atmospheric attenuation occurs due to particles that absorb or scatter the transmitted light. Specifically, particles that affect FSO communication systems include those associated with rain, snow, fog, pollution, smoke, and other environmental factors \cite{Atmos_Atten}. Atmospheric attenuation is typically modeled as a deterministic function of distance, as it is considered a fixed scaling factor over long periods of time, with no inherent randomness in its behavior.

Atmospheric turbulence distorts both the intensity and phase of the optical signal. This distortion is due to the fact that turbulence changes air pressure and temperature along the propagation path, causing random variations in the atmospheric refractive index~\cite{XuSensors2021}. In the literature, a large number of statistical models have been proposed to characterize the impact of atmospheric turbulence, particularly in terms of intensity fluctuations, which are a key challenge in intensity-modulation/direct-detection (IM/DD) systems. These mathematical models can be classified according to their regions of validity.
In the case of weak turbulence, the log-normal distribution has been proposed to model intensity fluctuations \cite{LOGNORMAL}. In moderate and strong atmospheric turbulence regimes, the Gamma-Gamma distribution and the K-distribution are commonly considered~\cite{AlH01}. 
In this paper, we focus on IM/DD systems and the weak turbulence regime assuming a log-normal distribution.

The third and fourth effects under consideration are {geometric} spread and pointing errors. Geometric spread refers to the natural divergence of a wave as it propagates through space. As the wave moves away from the transmitter, its energy spreads over a larger area, reducing the power density at the receiver.
Pointing errors are caused by misalignment of the transmit beam with respect to the field of view of the receiver. This misalignment is often caused by meteorological phenomena such as wind loads
and thermal expansions. 
As a result, random movements occur in the relative position of the transmitter and receiver. 
FSO systems reduce pointing errors by using active link alignment systems and acquisition, tracking, and pointing mechanisms~\cite{Korevaar03,Tracking}.
In this work, we consider small pointing errors, which are expected to be achieved by future FSO systems.

In order to assess the performance of FSO communication systems, multiple works (e.g., \cite{Farid2007, Boluda2017, 7501584}) have modeled the combined impact of atmospheric attenuation, atmospheric turbulence, geometric spread, and pointing errors. 
The statistical model presented in \cite{Farid2007} is one of the most widely used to characterize optical intensity fluctuations at the receiver in the weak turbulence regime. This model provides the joint probability density function (PDF) that describes the complete channel response determined by the four phenomena considered in this paper~\cite[Sec.~III-D]{Farid2007}. 
This PDF is expressed using the complementary error function. The integral form of this function makes it difficult to find simple expressions of important performance metrics such as the bit error rate (BER) and the symbol error rate (SER). 
The exact average BER and SER expressions for an FSO communication system using $M$-ary pulse amplitude modulation ($M$-PAM) involve a double integral that is complex to evaluate.

Recently, we proposed in~\cite{ICTON, OECC} easy-to-compute approximations for the average BER for on-off keying (OOK). 
Our proposed expressions in~\cite{ICTON, OECC} are limited to OOK and do not always provide accurate results. 
In this paper, we extend our previous works~\cite{ICTON, OECC} in three different ways, leading to the following three main contributions of this paper.
First, we derive an exact expression of the average SER for $M$-PAM formats, thereby extending our previous work on BER to SER, and also from OOK to high-order modulation formats. 
Second, we propose a new simplified approximate SER expression for $M$-PAM. Our approximation is a sum of two one-dimensional integrals and is easier to compute than the exact one.
We show that our proposed SER approximation for $M$-PAM is in excellent agreement with the exact expression, and thus with the Monte Carlo simulations. To the best of our knowledge, our expression is the most accurate approximation for the average SER for an FSO channel affected by weak atmospheric turbulence, geometric spread, pointing errors, and atmospheric losses.
Lastly, we demonstrate the usefulness of the developed approximations by performing two asymptotic analyses. First, we investigate the additional transmit power required to maintain the same SER when the spectral efficiency increases by $1$ bit/symbol. Second, we study the asymptotic behavior of our SER approximation for dense PAM constellations and high transmit power.

The paper is organized as follows. In Sec.~\ref{Sec:SystemModel}, we present the system model considered in this work. In Sec.~\ref{Sec:effect}, the main channel-induced impairments for the atmospheric link are discussed, and the distribution of the composite channel gain coefficient is presented. The main contributions of this work are presented in Sec.~\ref{Sec:SER}. Finally, numerical results and conclusions are presented in Secs.~\ref{sec:results} and~\ref{sec:Conclusions}, resp.

 \vspace{-1ex}
 \section{System Model}
\label{Sec:SystemModel}

We consider an FSO system that employs IM and DD with $M$-PAM modulation, as illustrated in Fig.~\ref{fig:system}. Compared to OOK, $M$-PAM transmits $m$ bits per symbol, where $m=\log_2 M$, which offers higher spectral efficiencies. IM is realized by first converting vectors of $m$ independent and uniformly distributed information bits $\boldsymbol{b}=(b_1,b_2,\dotsc,b_{m})$ into an $M$-PAM symbol $x_{i}$ at time instant $i$, where $x_{i} \in \mathcal{X} $, and where \mbox{$\mathcal{X} \triangleq \{j2\hat{X}/(M-1): j=0,1,\dotsc,M-1\}$}.\footnote{Notation convention: Bold symbols denote vectors, capital letters denote random variables (RVs), calligraphic letters denote sets, and $\mathbb{E}[\cdot]$ denotes expectation.}  
The bit-to-symbol mapping follows the binary reflected Gray code (BRGC)~\cite{BRGC}. 
The $M$-PAM symbols are then interleaved by $\Pi$ and passed through an ideal filter that generates an electrical signal. This electrical signal is then used to modulate the light source. The average optical transmitted power is defined as $P\triangleq \alpha \mathbb{E}[X]$, where $\alpha$ is the electro-optical conversion factor. In this work, we assume $\alpha=1$ (W/A), and thus it can be shown that the power constraint is satisfied for $\mathcal{X}$ if $\hat{X}=P$. From now on, we consider \mbox{$\mathcal{X} \triangleq \{j2P/(M-1): j=0,1,\dotsc,M-1\}$}. 

The optical signal is then transmitted through the FSO channel. At the receiver side, a photodetector (PD) detects the power of the signal. The photocurrent generated by the PD is directly proportional to the incident optical power via the detector responsivity $\beta$ (A/W). The resulting electrical signal is then passed through a receiver filter and sampled at the symbol rate $1/T$ to recover the transmitted symbols.
Finally, after deinterleaving by $\Pi^{-1}$, the vector of bit estimates $\boldsymbol{\hat{b}}=(\hat{b}_1,\hat{b}_2,\dotsc,\hat{b}_{m})$ is obtained by making hard decisions (HD) on the symbols $y_{i}$. 


The input-output relationship between symbols $x_i$ and $y_i$ (see Fig.~\ref{fig:system}) for an FSO link can be modeled as~\cite[eq.~(1)]{Farid2007}
\begin{equation}
     y_{i}= \eta h_i x_{i} +n_{i},
    \label{ChannelEq}
\end{equation}
where $\eta=\alpha\beta$, $h_{i}$ is the channel gain coefficient, and $n_{i}$ is a zero-mean additive white Gaussian noise sample with variance $\sigma_n^2$. 
The coherence time of the channel and the time scales of the fading processes (on the order of milliseconds,~\cite{SlowFading},~\cite[Sec.~II]{Farid2007}) are much larger than the duration of the bit (in the order of nanoseconds for Gbps rates \cite{Farid2007}). Therefore, the channel gain coefficient $h_i$ in \eqref{ChannelEq} is constant over a large number of transmitted symbols, which effectively results in a slow fading channel. However, in this paper, we assume an infinitely long interleaver $\Pi$, and thus the input-output relationship between symbols $x_i$ and $y_i$ is modeled via~\eqref{ChannelEq} under the assumption that $h_i$ are i.i.d. random variables. This assumption causes each received symbol to experience a different fading realization.

Based on the discussion above, we conclude that the channel model in \eqref{ChannelEq} is memoryless, and thus, from now on, we omit the index $i$. The real-valued channel gain coefficient $h$ in \eqref{ChannelEq} models the random attenuation of the signal intensity in the propagation channel. In our model, this attenuation is due to four factors: atmospheric loss ($h_{l}$), geometric spread ($h_{g}$), pointing errors ($h_{p}$), and atmospheric turbulence ($h_{a}$). Since these phenomena create independent multiplicative losses, $h$ can be formulated as~\cite[eq.~(2)]{Farid2007} 
\begin{equation}
    H=h_{l}h_{g}H_{a}H_{p}.
    \label{Coeff_H}
\end{equation}
For a given distance $z$, $h_{l}$ and $h_{g}$ are assumed to be constant. On the other hand, $H_{a}$ and $H_{p}$ are RVs whose distributions will be discussed in the next section.

In this work, we consider maximum likelihood detection with perfect channel state information, i.e., both $h$ and $\eta$ in~\eqref{ChannelEq} are assumed to be known at the receiver.
We distinguish between two definitions of the received signal-to-noise ratio (SNR): the average received \emph{electrical} SNR, which is commonly used in the field of information and communication theory, and the received \emph{optical} SNR which is sometimes considered a more appropriate metric for FSO links, as described in~\cite[Sec.~11.2.1]{And05}. 

The average received electrical SNR is defined as 
\begin{align}
 \label{SNR_el_av0}
    \overline{\text{SNR}}_\text{e} &\triangleq \frac{\eta^2  \mathbb{E}[X^2]\mathbb{E}[H^2]}{\sigma_n^2} \\
    \label{SNR_el_av1}
    &= \frac{\eta^2 \sum_{j=0}^{M-1} \left(\frac{j2P}{M-1}\right)^2 h_{l}^2 h_{g}^2\mathbb{E}[H_{a}^2]\mathbb{E}[H_{p}^2]}{\sigma_n^2}\\
    &= \frac{\eta^2 \frac{2P^2 \cdot M(2M - 1)}{3(M - 1)}  h_{l}^2 h_{g}^2\mathbb{E}[H_{a}^2]\mathbb{E}[H_{p}^2]}{\sigma_n^2} ,
    \label{SNR_el_av}
\end{align}
where $\mathbb{E}[X^2]$ represents the average electrical transmitted power. For $M$-PAM constellations considered in this paper, \eqref{SNR_el_av} shows that the average electrical transmitted power varies with the modulation order $M$.

We conclude this section by defining the average received optical SNR as in \cite[Sec.~11.2.1]{And05}\footnote{The optical SNR defined in this work should not be confused with the optical SNR (OSNR) typically used in optical fiber networks, which quantifies the ratio of signal power to optical noise power within a given optical bandwidth \cite[eq.~(33)]{OSNR}.}
\begin{equation}
    \overline{\text{SNR}}_\text{o} \triangleq\frac{\eta \mathbb{E}[X] \mathbb{E}[H]}{\sigma_n}=\frac{\eta  Ph_{l} h_{g}\mathbb{E}[H_{a}]\mathbb{E}[H_{p}]}{\sigma_n}.
    \label{SNR_op_av}
\end{equation}
The definition in \eqref{SNR_op_av} reflects the ratio between the detector signal photocurrent and the noise standard deviation $\sigma_n$. Unlike $\overline{\text{SNR}}_\text{e}$ in \eqref{SNR_el_av}, the right-hand side of \eqref{SNR_op_av} shows that $\overline{\text{SNR}}_\text{o}$ remains constant across different modulation formats. Both \eqref{SNR_el_av} and \eqref{SNR_op_av} show the importance of the first and second moments of $H_{a}$ and $H_{p}$. We will discuss these in Sections~\ref{subSec_turbulence} and \ref{subSec_pointing} below.


 \vspace{-1ex}
 \section{Optical Channel Fading Model}
 \label{Sec:effect}

  \subsection{Atmospheric Attenuation}
 
The propagation of light through the atmosphere is affected by absorption and scattering. 
These phenomena cause the loss of a part of the transmitted power. 
This loss can be modeled via the well-known Beer-Lambert law given in~\cite[Sec. III-B]{Farid2007} and~\cite[eq.~(1)]{hl} as
\begin{equation}
    h_{l} = \exp{(-\sigma(\lambda) z)},
    \label{eq:hl}
\end{equation}
where $\sigma(\lambda)$ is a wavelength-dependent attenuation coefficient (dB/km). 
This coefficient depends on weather conditions and can be obtained from atmospheric visibility according to Kim's model~\cite[eq.~(3)]{Kim}. Since the attenuation coefficient slowly changes over time, in our analysis we assume that $h_l$ is a deterministic coefficient.

 \subsection{Atmospheric Turbulence}
 \label{subSec_turbulence}
 
 In order to develop mathematical models to characterize the signal fluctuations caused by atmospheric turbulence $H_{a}$ in a wide range of atmospheric conditions, numerous studies have been carried out (see, e.g., \cite{AlH01,Jur11}, and references therein). In this work, we consider the weak turbulence regime. For weak turbulence, the PDF of $H_{a}$ is typically modeled using a log-normal distribution \cite{AlH01} which is given by \cite[eq.~(1)]{Ansari2014} and \cite[eq.~(16)]{Niu12} as
\begin{equation}   
    f_{H_{a}}(h_{a})=\frac{1}{h_{a}\sqrt{2\pi\sigma^2}}\exp{\left(-\frac{(\hbox{ln}(h_{a})-\delta)^2}{2\sigma^2}\right)}, \hspace{1ex} h_a>0, 
    \label{eq:fha_LN}
\end{equation}
where $\delta$ is the log-scale parameter \cite{Ansari2014} and $\sigma^2$ is the log-variance of $H_a$, which is directly linked to amount of turbulence in the channel.

In this work, we set $\delta=-\sigma^2$ in order to normalize the second moment of the atmospheric turbulence to one, i.e., 
\begin{equation}
\label{norm.Ha}
\mathbb{E}[H_{a}^2]=\exp{\left( 2\delta + 2\sigma^2\right)}=1.
\end{equation}
Although the normalization in \eqref{norm.Ha} is not commonly used in the FSO literature, we adopt it here to maintain a constant average \emph{electrical} SNR (see \eqref{SNR_el_av}) when the modulation format is fixed, even if turbulence conditions change.\footnote{This is for a given value of $\mathbb{E}[H_{p}^2]$.} This normalization ensures a fair comparison of channels with different turbulences, which we analyze in this paper, and follows the standard practice in wireless communications \cite[Sec.~III-C]{Proakis_Paper}. 

Under the normalization in \eqref{norm.Ha} obtained using $\delta=-\sigma^2$, the first moment of $H_{a}$ is given by
\begin{equation}\label{Ha.1st}
\mathbb{E}[H_{a}]=\exp{\left(\delta +\sigma^2/2\right)}=\exp{\left(-\sigma^2/2\right)},
\end{equation}
which shows that as $\sigma^2$ increases, the optical SNR in \eqref{SNR_op_av} decreases (again, for a fixed value of $\mathbb{E}[H_{p}^2]$). This is the expected behavior: the higher turbulence, the lower the average optical SNR.\footnote{This normalization differs from the standard normalization used in the literature (see, e.g., \cite{Farid2007,XuSensors2021,7501584}), where $\mathbb{E}[H_{a}]$ is normalized to 1. The normalization used in the literature leads to a value of $\mathbb{E}[H_{a}^2]$ that is proportional to $\sigma^2$, which in turn results in an unexpected result: the higher the turbulence, the higher the average electrical SNR.}
 
Under weak turbulence regime $\sigma^2\approx\sigma_R^2$ and $\sigma_R^2<1$~\cite[Fig.~7.4]{And05}, where $\sigma_R^2$ is the Rytov variance. The Rytov variance measures the degree of fluctuation of the optical signal intensity due to turbulence effects and is defined as
\begin{equation}
    \sigma_R^2=1.23 C_n^2(L)k^{7/6}z^{11/6}.
    \label{Rytov}
\end{equation}
In \eqref{Rytov}, $C_n^2(L)$ is the index of refraction structure parameter at altitude $L$ from ground level (assumed to be constant along the propagation path for horizontal links) and $k=2\pi/\lambda$ is the optical wave number, where $\lambda$ represents the wavelength of the optical signal. Note that $\sigma_R^2$ in \eqref{Rytov} can be measured directly from atmospheric parameters, such as temperature, pressure, and humidity, which influence the refractive index and, consequently, $C_n^2(L)$. 
By varying $\sigma_R^2$, and consequently $\sigma^2$, we can control the level of turbulence in the considered system.


 \subsection{Geometric Spread and Pointing Errors}
 \label{subSec_pointing}

In FSO links, in addition to atmospheric turbulence, geometric spread and pointing accuracy also affect the performance of the system. The geometric spread is caused by the divergence of the transmit beam when propagating through the atmosphere. This spread is modeled via $w_z \approx \theta \cdot z/2$, where $w_z$ is the received beam waist radius and $\theta$ is the transmitter divergence angle, as shown in Fig. \ref{fig:system}. Since the received beam width is usually wider than the PD aperture size $a$, part of the transmitted power cannot be collected, leading to losses. Thus, this geometric spread depends mainly on the ratio between the received beam waist $w_z$ and the receiver aperture radius $a$. 
On the other hand, pointing error effects produced by wind loads and thermal expansions result in building sways leading to a misalignment between transmitter and receiver. 
Additionally, pointing errors worsen due to imperfections in the beam pointing, acquisition, and tracking processes. 

Geometric spread and pointing accuracy are interrelated effects, and thus they are usually jointly modeled in the literature\cite{Farid2007, Boluda2017, Pointing_Geomet}. 
However, in this work, without loss of generality, we model them separately using the general framework proposed in~\cite{Farid2007}. 
Specifically,~\cite{Farid2007} assumes a beam profile with beam waist $w_{z}$ on the receiver plane, and a circular aperture receiver of radius $a$, as shown in Fig.~\ref{Fig:hp}. Under these assumptions, we present models for each effect.

\subsubsection{Geometric Spread Model}
Geometric spread affects the fraction of
optical power collected by the PD. More precisely, $h_{g}$ represents the fraction of the collected power in the absence of pointing errors, i.e., when the radial displacement $r$ between the beam center and the PD aperture center is equal to zero. This corresponds to the case shown in Fig.~\ref{Fig:hp} (left). The attenuation due to the geometric spread is defined as~\cite[Sec.~III-C]{Farid2007} 
\begin{equation}
    \label{eq:hg}
h_{g}\triangleq\mathrm{erf}(v_0)^2,
\end{equation}
where $\mathrm{erf}(\cdot)$ is the error function, and
\begin{equation}
    v_0\triangleq \frac{\sqrt{\pi}a}{\sqrt{2}w_z}.
\end{equation}
For a given distance $z$, the attenuation due to geometric spread $h_{g}$ is constant.

\subsubsection{Pointing Error Model}
Pointing errors affect the alignment between the received optical beam and the PD aperture, leading to additional power loss. More precisely, $H_{p}$ represents the fraction of collected power when a misalignment occurs, i.e., when $r>0$ (see Fig.~\ref{Fig:hp} (right)). 
The attenuation due to pointing errors can be approximated as~\cite{Farid2007} 
\begin{equation}
    H_p \approx \kappa \exp{\left( -\frac{2R^2}{\hat{w}_{z}^2}\right)} ,
    \label{hp}
\end{equation}
where $R$ is the (random) radial displacement and $\hat{w}_{z}^2$ is the equivalent beam width. Here $\hat{w}_{z}^2$ is given by \cite[eq.~(9)]{Farid2007}
\begin{equation}
    \hat{w}_{z}^2\triangleq \frac{w_z^2 \sqrt{\pi}\mathrm{erf}(v_0)}{2 v_0 \exp(-v_0^2)}.
\end{equation}
In analogy to the atmospheric turbulence analysis, we included a normalizing parameter $\kappa$ in \eqref{hp} defined as
\begin{equation}
    \label{eq:Kappa}
    \kappa\triangleq\sqrt{\frac{\gamma^2+2}{\gamma^2}}.
\end{equation}
This normalization parameter results in a normalized second moment for $H_p$, as we will see below.

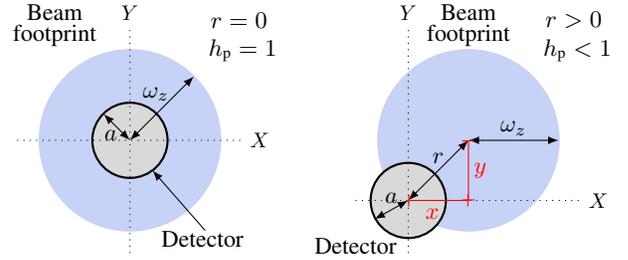
\begin{figure}[t]
        \centering
        \definecolor{RoyalBlue}{rgb}{0.25, 0.41, 0.88}

        \begin{tikzpicture}[scale=1]

        \filldraw[fill=RoyalBlue!30, draw=RoyalBlue!30, thick] (-1,0) circle (1.2);  
           \filldraw[fill=gray!30, draw=black, thick] (-1,0) circle (0.5);  

             \filldraw[fill=RoyalBlue!30, draw=RoyalBlue!30, thick] (3.5,0) circle (1.2);  
           \filldraw[fill=gray!30, draw=black, thick] (2.7,-0.8) circle (0.5);  %

              \node at (0.45, 1.6) {\small ${r=0}$};
             \node at (0.5, 1.2) {\small ${h_\text{p}=1}$};
              \node at (-2,1.7) {\small Beam};
            \node at (-2, 1.4) {\small footprint};
            
             \node at (4.9, 1.6) {\small ${r>0}$};
             \node at (4.95, 1.2) {\small ${h_\text{p}<1}$};
              \node at (3.5,1.7) {\small Beam};
            \node at (3.5, 1.4) {\small footprint};
             
            \draw[-, dotted,thin] (-2.5,0) -- (0.5,0) node[right,scale=0.8] {$X$};
            \draw[-, dotted,thin]  (-1,-1.5) -- (-1,1.5) node[above,scale=0.8] {$Y$};

            \draw[-, dotted,thin] (2,-0.8) -- (5,-0.8) node[right,scale=0.8] {$X$};
            \draw[-, dotted,thin]  (2.7,-1.5) -- (2.7,1.5) node[above,scale=0.8] {$Y$};

 
            \draw[latex-latex,thin] (-1,0) -- (-0.15,0.85)node[midway,yshift=1.2ex, xshift=-0.5ex,scale=0.9] {\normalsize $\omega_z$};
            
            \draw[latex-latex,thin] (-1,0) -- (-1.35,0.35)node[midway,yshift=-0.7ex, xshift=-0.5ex,scale=0.9] {\normalsize $a$};

              \draw[-latex,thin] (0,-1.2) -- (-0.7,-0.4)node[midway,yshift=-3.2ex, xshift=2ex,scale=1] {\small Detector};

            \draw[latex-latex,thin] (2.7,-0.8) -- (2.25,-1.05)node[midway,yshift=1ex,xshift=0ex,scale=0.9] {\normalsize$a$};

            \draw[latex-latex,thin] (2.7,-0.8) -- (3.5,0)node[midway,yshift=1ex,xshift=0ex,scale=0.9] {\normalsize ${r}$};

            \draw[latex-latex,thin] (3.5,0) -- (4.7,0)node[midway,yshift=1ex,scale=0.9] {\normalsize $\omega_z$};

            \draw[|-|,thin, color=red] (2.7,-0.8) -- (3.5,-0.8)node[midway,below,xshift=-0.5ex,scale=0.9] {\normalsize $x$};

            \draw[|-|,thin, color=red] (3.5,-0) -- (3.5,-0.8)node[midway,xshift=1ex,scale=0.9] {\normalsize $y$};

            \node at (2,-1.4) {\small Detector};

        \end{tikzpicture}
    \caption{Detector (gray) and beam footprint (blue) without misalignment (left) and with misalignment (right) on the detector plane.}
   \label{Fig:hp}
\end{figure}

We assume i.i.d. Gaussian-distributed horizontal $x$ and vertical $y$ displacements in the receiver plane, see Fig.~\ref{Fig:hp} (right), and thus the radial error $r=\sqrt{x^2+y^2}$ is modeled via a Rayleigh distribution 
\begin{equation}
    f_R(r)=\frac{r}{\sigma_s^2}\exp{\left( -\frac{r^2}{2\sigma_s^2}\right)}, \quad r>0,
    \label{RadialPDF}
\end{equation}
where $\sigma_s^2$ is the jitter variance at the receiver. The parameter $\sigma_s$ is estimated as $\sigma_s\approx\theta_s z$, where $\theta_s$ is defined as the jitter angle. By varying $\theta_s$, $\sigma_s$ can be adjusted, which allows us to control the level of pointing errors in our model.

Under the assumptions above, the channel attenuation due to pointing errors $H_{p}$ can be seen as a function of the radial displacement $R$. In this work, the PDF of the pointing errors is based on the formulation in \cite[Sec.~III-C]{Farid2007}, which leads to
\begin{equation}
    f_{H_{p}}(h_{p})= \frac{\gamma ^2}{\kappa^{\gamma ^2}} h_{p}^{\gamma^2-1}, \quad  0\leq h_{p} \leq \kappa,
    \label{pdf:hp}
\end{equation} 
where 
\begin{equation}
\label{eq:gamma}
\gamma 
\triangleq \frac{\hat{w}_{z}}{2\sigma_s},
\end{equation}
 and where $\kappa$ is given in \eqref{eq:Kappa}.
In \eqref{eq:gamma}, $\gamma$ denotes the ratio between the equivalent beam radius at the receiver and the standard deviation of the pointing error displacement. In analogy with the normalization applied in the atmospheric turbulence modeling, we incorporated here the normalizing parameter $\kappa$ in \eqref{pdf:hp}, which results in a normalized second moment for $H_p$, i.e.,
\begin{equation}
\label{norm.Hp}
    \mathbb{E}[H_{p}^2]=\kappa^2\frac{\gamma^2}{\gamma^2+2}=1.
\end{equation} 
Under the normalization in \eqref{norm.Hp}, the first moment of $H_{p}$ is given by
\begin{equation}\label{Hp.1st}
\mathbb{E}[H_{p}]=\kappa\frac{\gamma^2}{\gamma^2+1}.
\end{equation}
In view of \eqref{eq:gamma}, \eqref{Hp.1st} shows that as $\sigma_s^2$ increases, $\mathbb{E}[H_{p}]$ decreases, leading to a decrease in the optical SNR in \eqref{SNR_op_av}. This is again the expected behavior: the higher the severity of pointing errors, the lower the average optical SNR.


 \subsection{Composite Channel Model}

In the weak turbulence regime, the distribution of the composite channel gain coefficient \hbox{$H=h_{l}h_{g}H_{a}H_{p}$}, where $f_{H_{a}}(h_{a})$ and $f_{H_{p}}(h_{p})$ are given by~\eqref{eq:fha_LN} and~\eqref{pdf:hp}, resp., can be expressed as
 \begin{align}
     f_{H}(h)&=\frac{\gamma^2h^{\gamma^2 -1}}{2(h_{g} h_{l} \kappa)^{\gamma^2}}  \hbox{erfc} (v) \exp{\left[ \gamma^2\sigma_R^2 \left(1+\gamma^2/2\right)\right]},
     \label{eq:fh_LN_sim}
   \end{align}
where $h_{l}$ is given by \eqref{eq:hl}, $\sigma_R^2$ by \eqref{Rytov}, $h_{g}$ by \eqref{eq:hg}, $\gamma$ by \eqref{eq:gamma},
 \begin{equation}
v \triangleq\frac{\mathrm{ln}\bigl(\frac{h}{h_lh_g \kappa}\bigr)+\mu}{\sqrt{2\sigma^2}},
\label{psi}
\end{equation}
 \begin{equation}
\mu \triangleq \sigma_R^2(\gamma^2+1), 
\label{eq:mu}
\end{equation}
and the complementary error function is defined as
 \begin{equation}
\mathrm{erfc}(z)\triangleq\frac{2}{\sqrt{\pi}}\int_{z}^{\infty} \exp{(-t^2)} dt = 2 Q\left(\sqrt{2}z\right),
\label{eq:erfc}
\end{equation}
where $Q(\cdot)$ denotes the Gaussian Q-function.

Fig.~\ref{FigPDF} shows the PDF of $H$ in \eqref{eq:fh_LN_sim} for nine different operating points, obtained by jointly varying the impact of atmospheric turbulence and pointing errors. 
The severity of turbulence is constrained to $\sigma_R^2\leq 1$, which corresponds to the weak turbulence regime~\cite[Fig.~7.4]{And05}. 
The severity of pointing errors is varied from the mildest pointing error regime $\sigma_s=0.35$~m (corresponding to a jitter angle of $\theta_s=0.116$) to the strictest pointing error regime $\sigma_s=0.2$~m (corresponding to a jitter angle of $\theta_s=0.067$)~\cite{Korevaar03,Boluda2017}. 
Specifically, three levels of atmospheric turbulence are considered, $\sigma_R^2=0.1$, $\sigma_R^2=0.5$, and $\sigma_R^2=0.9$, shown in Fig.~\ref{FigPDF} with pink, orange, and green, resp. In terms of pointing errors, we consider $\sigma_s=0.35$~m, $\sigma_s=0.25$~m, and $\sigma_s=0.2$~m represented with solid, dashed, and dotted curves, resp.
We consider the attenuation coefficient \mbox{$\sigma(\lambda)=0.2208$~(dB/km)} corresponding to clear air weather condition. The remaining simulation parameters are summarized in Table~\ref{tab:System_params} (Sec.~\ref{sec:results}, ahead).
Table \ref{tab:Values_E_h} presents the numerical value of $\mathbb{E}[H]$ corresponding to each operating point. A higher $\mathbb{E}[H]$ implies a higher average optical SNR (see~\eqref{SNR_op_av}).
Among all the operating points, the green solid curve in Fig.~\ref{FigPDF}—which corresponds to the combination of the highest severity of turbulence ($\sigma_R^2=0.9$) and the mildest pointing error regime ($\sigma_s=0.35$~m)—exhibits the lowest mean channel gain, i.e., the smallest value of the first moment $\mathbb{E}[H]$ (see Table~\ref{tab:Values_E_h}), and consequently the lowest average optical SNR.\footnote{Note that due to the normalization in \eqref{norm.Ha} and \eqref{norm.Hp}, the second moment of the channel gain coefficient $\mathbb{E}[H^2]$ remains constant across the nine operating points.} This scenario represents the most severe channel degradation among those considered, as it combines the strongest effects of atmospheric turbulence and pointing errors. 

Furthermore, both Fig. \ref{FigPDF} and Table \ref{tab:Values_E_h} show that increasing the severity of the pointing errors (from $\sigma_s=0.2$~m to $\sigma_s=0.35$~m) has only a minor impact in the first moment of the channel gain coefficient $\mathbb{E}[H]$, and consequently on the average optical SNR. This observation suggests that atmospheric turbulence is the dominant factor in the performance degradation observed in this study.

\begin{figure}[tbp!]
        \centering
            \usetikzlibrary{arrows,decorations.markings}
    \usetikzlibrary{calc,arrows}    
    
    \definecolor{my_pink}{rgb}{1,0.07,0.65}
\definecolor{my_g}{rgb}{0.09,0.71,0.14}
\definecolor{FEC}{rgb}{0.5, 0.0, 0.13}
    \begin{tikzpicture}[scale=1]    
    
   \begin{semilogyaxis}[ 
       width=0.48\textwidth,
       height=3.1in,   
        xmin=1e-6, xmax=1.6e-3,
        ymin=1e1, 
        ymax=2.5e3,  
        xlabel shift=-1ex,
        xlabel={$h$}, 
        ylabel shift=-1ex,
        ylabel={$f_H(h)$},
        grid=major
    ]
    \addplot [color=black,solid,line width=1pt,mark options={solid}]coordinates {
             (1e-6,0.4) (1e-6,0.5)};\label{True}
    \addplot [color=black,dashed ,line width=1pt,mark options={solid}]coordinates {
             (1e-6,0.4) (1e-6,0.5)};\label{Dashed}
    \addplot [color=black,dotted,line width=1pt,mark options={solid}]coordinates {
             (1e-6,0.4) (1e-6,0.5)};\label{Aprox1}


    
     \addplot [color=orange,dashed,line width=1pt,mark options={solid}]file {./txtData/PDF/Orig_varR05_jit25.txt};

       \addplot [color=my_pink,solid,line width=1pt,mark options={solid}]file {./txtData/PDF/Orig_varR01_jit35.txt};
    
       \addplot [color=my_g,dotted,line width=1pt,mark options={solid}]file {./txtData/PDF/Orig_varR09_jit2.txt};

       \addplot [color=my_g,dashed,line width=1pt,mark options={solid}]file {./txtData/PDF/Orig_varR09_jit25.txt};
       \addplot [color=my_g,solid,line width=1pt,mark options={solid}]file {./txtData/PDF/Orig_varR09_jit35.txt};
       
       \addplot [color=orange,solid,line width=1pt,mark options={solid}]file {./txtData/PDF/Orig_varR05_jit35.txt};
       
       \addplot [color=orange,dotted,line width=1pt,mark options={solid}]file {./txtData/PDF/Orig_varR05_jit2.txt};

       \addplot [color=my_pink,dashed,line width=1pt,mark options={solid}]file {./txtData/PDF/Orig_varR01_jit25.txt};
       
       \addplot [color=my_pink,dotted,line width=1pt,mark options={solid}]file {./txtData/PDF/Orig_varR01_jit2.txt};

    \end{semilogyaxis}

\node [draw,fill=white,anchor= south west,font=\scriptsize] at (1.2,0.1) {\shortstack[l]{ 
\ref{True}  $\sigma_s=0.35$~m\\
\ref{Dashed} $\sigma_s=0.25$~m\\
\ref{Aprox1}  $\sigma_s=0.2$~m}};

    \node [coordinate](input) {};

    \draw[-stealth,solid, black,line width=0.7pt,rounded corners]($(input.north)+(4.4,4.5)$) -- ($(input.north)+(3.3,4)$);
\node at ($(input.north) + (5.3, 4.8)$) [scale=1.2, font=\scriptsize] {Increasing Impact};
\node at ($(input.north) + (5.3,4.5)$) [scale=1.2, font=\scriptsize] {of Turbulence};

\node at ($(input.north) + (4.7, 6.1)$) [scale=1.2, font=\scriptsize] {Increasing Impact};
\node at ($(input.north) + (4.7,5.8)$) [scale=1.2, font=\scriptsize] {of Pointing errors};
  \draw[-stealth,solid, black,line width=0.7pt,rounded corners]($(input.north)+(3.6,5.7)$) -- ($(input.north)+(3,5.3)$);
  
\node[ellipse, color=my_g, line width=0.7pt, draw,minimum width = 0.6cm, minimum height = 0.05cm,rotate=55] (e) at (2.5,4.8) {};

\node [draw,color=my_g, fill=white,anchor= south west,font=\scriptsize,text=black] at (1.5,3.9) {\shortstack[l]{
$\sigma_R^2=0.9$}};

\node[ellipse, color=orange, line width=0.7pt, draw,minimum width = 0.6cm, minimum height = 0.05cm,rotate=185] (e) at (0.45,4) {};

\node [draw,color=orange, fill=white,anchor= south west,font=\scriptsize,text=black] at (1.4,2.6) {\shortstack[l]{
$\sigma_R^2=0.5$}};
 \draw[-stealth,solid, black,line width=0.7pt,rounded corners]($(input.north)+(0.75,3.9)$) -- ($(input.north)+(1.35,3.1)$);

\node[ellipse, color=my_pink, line width=0.7pt, draw,minimum width = 0.6cm, minimum height = 0.05cm,rotate=185] (e) at (5.6,1.9) {};

\node [draw,color=my_pink, fill=white,anchor= south west,font=\scriptsize,text=black] at (3.8,1.7) {\shortstack[l]{
$\sigma_R^2=0.1$}};

    \end{tikzpicture}
    \caption{PDF of $H$ in \eqref{eq:fh_LN_sim} for nine different operating points ($\sigma_s$, $\sigma_R^2$).}
    \label{FigPDF}
\end{figure}
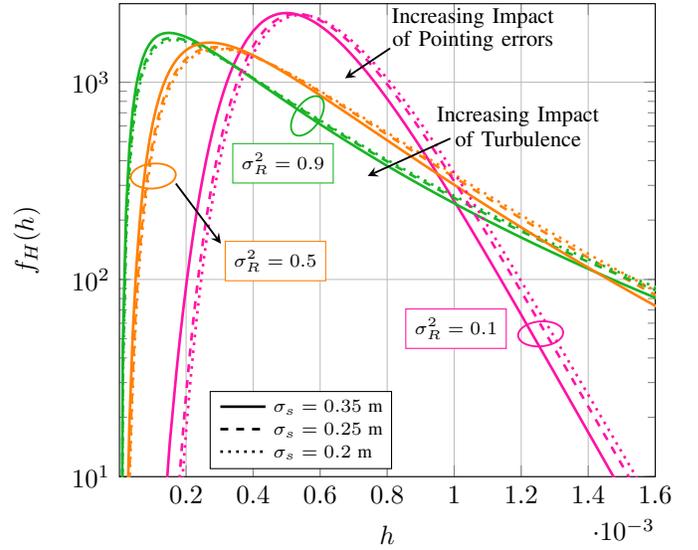

\begin{table}[t]
\renewcommand{\arraystretch}{1.3}
\centering
\caption{Numerical Values of $\mathbb{E}[H]$ for Each Operating Point ($\sigma_s$, $\sigma_R^2$).}
\begin{tabular}{cccc}
\hline
{\textbf{$\boldsymbol{\sigma_s}$}}& {\textbf{$\boldsymbol{\sigma_R^2}$}} & \textbf{Trace in Fig.~\ref{FigPDF}} &{\textbf{$\boldsymbol{\mathbb{E}[H]}$}} \\
\hline
\multirow{3}{*}{$0.35$~m}& $0.9$    & {\color{my_g}\rule[0.5ex]{2em}{1pt}}& $4.16\cdot10^{-4}$\\
 & $0.5$  & {\color{orange}\rule[0.5ex]{2em}{1pt}}  &  $5.09\cdot10^{-4}$\\
& $0.1$  & {\color{my_pink}\rule[0.5ex]{2em}{1pt}}    &   $6.21\cdot10^{-4}$\\
\hline
\multirow{3}{*}{$0.25$~m} & $0.9$  & \tikz[baseline]{\draw[dashed, color=my_g, line width=1pt] (0,0) -- (2em,0);}  &   $4.18\cdot10^{-4}$\\
& $0.5$    & \tikz[baseline]{\draw[dashed, color=orange, line width=1pt] (0,0) -- (2em,0);}  &  $5.11\cdot10^{-4}$\\
 &$0.1$   & \tikz[baseline]{\draw[dashed, color=my_pink, line width=1pt] (0,0) -- (2em,0);}  &   $6.24\cdot10^{-4}$\\
\hline
\multirow{3}{*}{$0.2$~m}& $0.9$    & \tikz[baseline]{\draw[dotted, color=my_g, line width=1pt] (0,0) -- (2em,0);} &  $4.19\cdot10^{-4}$\\
& $0.5$    & \tikz[baseline]{\draw[dotted, color=orange, line width=1pt] (0,0) -- (2em,0);}  &  $5.11\cdot10^{-4}$\\
&$0.1$    & \tikz[baseline]{\draw[dotted, color=my_pink, line width=1pt] (0,0) -- (2em,0);}  &  $6.25\cdot10^{-4}$\\
\hline
\end{tabular}
\label{tab:Values_E_h}
\end{table}

 \vspace{-1ex}
\section{Proposed Approximations}
\label{Sec:SER}

The PDF in \eqref{eq:fh_LN_sim} is not a closed-form expression, since it involves the integral in~\eqref{eq:erfc}. This integral hinders the development of closed-form expressions of related performance metrics such as BER and SER.
In this section, we first present approximate BER expressions, previously proposed in our works~\cite{ICTON,OECC} for OOK signaling over an FSO channel. We then develop a new approximate SER expression for $M$-PAM. 
Furthermore, we derive BER and SER expressions and assess their accuracy with respect to the exact expressions.
Finally, we demonstrate the usefulness of the developed approximations by performing two asymptotic analyses. First, we investigate the additional transmit power required to maintain the same (low) SER when the spectral efficiency increases by $1$ bit/symbol. Second, we study the asymptotic behavior of our SER approximation for dense PAM constellations and high transmit power.

\subsection{Error Probability Preliminaries}
The BER conditioned on $H=h$ for $M$-PAM is defined as
\begin{equation}
    P_{b,M} \triangleq \Pr\{\hat{B}_k\neq B_k|H=h\},
    \label{eq:BER_OOK_conditioned}
\end{equation}
where $\hat{B}_k$ is the estimated bit and $k=1,2,\ldots,m$ is the bit position in the symbol. 

The SER conditioned on $H=h$, is defined as
\begin{align}
    P_{s,M} \triangleq & \Pr \{ \hat{X}\neq X|H=h\}
    \label{eq:SER_conditioned}\\
    \nonumber
    =& \frac{2(M-1)}{M}Q\left( \frac{\eta P h}{(M-1)\sqrt{\sigma_n^2}}\right)\\
    =& \bigg(\frac{M-1}{M}\bigg)\hbox{erfc}\left( \frac{\eta P h}{(M-1)\sqrt{2\sigma_n^2}}\right),
    \label{Eq:CSER_PAM}
\end{align}
where $\hat{X}$ is the estimated symbol. To pass from \eqref{eq:SER_conditioned} to \eqref{Eq:CSER_PAM}, we apply the traditional SER formula for $M$-PAM over AWGN, which uses the minimum distance between symbols and the Gaussian Q-function to model error probability. Incorporating the channel gain $h$ into the traditional SER expression for $M$-PAM leads to \eqref{Eq:CSER_PAM}.


The average BER $\overline{P}_{b,M}$ and the average SER $\overline{P}_{s,M}$ can be obtained by integrating their conditional 
expressions over the PDF in~\eqref{eq:fh_LN_sim}, i.e.,  
\begin{align}
\overline{P}_{b,M}=&\int_{0}^{\infty} P_{b,M}f_H(h)\, dh,
\label{eq:averageBER_PDF}\\
\overline{P}_{s,M}=&\int_{0}^{\infty} P_{s,M}f_H(h)\, dh.
\label{eq:averageSER_PDF}
\end{align}

For $M$-PAM, the average BER can also be approximated as~\cite[eq.~8]{BERAprox}
\begin{equation}
    \overline{P}_{b,M} \approx\frac{\overline{P}_{s,M}}{m},
    \label{eq:BER_MPAM}
\end{equation} 
where $\overline{P}_{s,M}$ is the average SER. The approximation in \eqref{eq:BER_MPAM} is known to be accurate at high SNR when the BRGC we consider in this paper is used.

We conclude this section by showing two different approximations for the \hbox{erfc}($z$) function in \eqref{eq:erfc} that we will use later on. The first approximation is
 \begin{align}
 \hbox{erfc}(z)\approx & 
   \begin{cases}
    \dfrac{2}{\sqrt{\pi}} \dfrac{\exp(-z^2)}{z+\sqrt{z^2+\frac{4}{\pi}}}, & \text{if}~z\geq0.\\ 
      1+\dfrac{\exp{\left(-\frac{2\pi z}{\sqrt{6}}\right)}-1}{\exp{\Big(-\frac{2\pi z}{\sqrt{6}}\Big)}+1}, &  \text{if}~z<0.
     \end{cases}
     \label{eq:Erfc_Aprox2}
 \end{align}
In \eqref{eq:Erfc_Aprox2}, we used the asymptotically tight (as $z\rightarrow \infty$) approximation in \cite[eq.~(7.1.13)]{QFun} for $z \geq 0$. 
For $z<0$, we used the asymptotically tight (as \hbox{$z\rightarrow -\infty$}) approximation given in \cite{QFun2}. These two approximations are shown in Fig.~\ref{FigERFC}, with green and orange circles, resp.
\begin{figure}[t]
        \centering
        \usetikzlibrary{arrows,decorations.markings}
\usetikzlibrary{calc,arrows}

\begin{tikzpicture}[scale=1]    

\definecolor{my_pink}{rgb}{1,0.07,0.65}
\definecolor{my_g}{rgb}{0.09,0.71,0.14}
\definecolor{FEC}{rgb}{0.5, 0.0, 0.13}

\definecolor{fuchsia}{rgb}{0.6,0.4,0.8}

\begin{axis}[ 
     width=0.49\textwidth,
    height=3.1in,   
    xmin=-3, xmax=4, 
   ymin=1e-4, ymax=2.1,
    ymode=log,
    xlabel={$z$}, 
    ylabel shift=-1ex,
    xlabel shift=-1ex,
    grid=major
]

\addplot [color=black,solid,line width=1pt,mark options={solid}]file {./txtData/Erfc_fun/ERFC.txt};\label{Q}

\addplot[only marks, color=blue, mark=*,mark options={solid,scale=0.7}]file{./txtData/Erfc_fun/Aprox.txt};

\addplot[only marks, color=orange, mark=*,mark options={solid,scale=0.7}]file{./txtData/Erfc_fun/Neg.txt};

\addplot[only marks, color=my_g, mark=*,mark options={solid,scale=0.7}]file{./txtData/Erfc_fun/UpperBound.txt};


\addplot [color=red,solid,line width=1pt,mark options={solid}]coordinates {   (0,1e-6) (0,60)};

\end{axis} 

\node at (0,0.12)(input1) {};


\node at ($(input1.south)+(6,5.7)$)[font=\small] {$\dfrac{1}{\sqrt{\pi}z}  \exp{(-z^2)}$}; 
\draw[->] ($(input.south)+(4.8,5.7)$)-- ($(input.south)+(3.85,5.6)$);  

\node at ($(input1.south)+(4.1,0.7)$)[font=\small] {$1+\dfrac{\exp{\left(-\frac{2\pi z}{\sqrt{6}}\right)}-1}{\exp{\left(-\frac{2\pi z}{\sqrt{6}}\right)}+1} $}; 
\draw[->] ($(input.south)+(5.5,1.2)$)-- ($(input.south)+(6,1.6)$);

\node at ($(input1.south)+(1.9,2.9)$)[font=\small] {$\dfrac{2}{\sqrt{\pi}}  \dfrac{\exp{\left(-{z^2}\right)}}{z+\sqrt{z^2+\frac{4}{\pi}}}$}; 
\draw[->] ($(input.south)+(1.4,3.5)$)-- ($(input.south)+(1.15,4.1)$); 

\node [draw,fill=white,anchor= south west,font=\small] at ($(input1.south)+(0.1,0.1)$) {\shortstack[l]{\ref{Q} \hbox{erfc}$(z)$ }};

\node [coordinate](input) {};

\end{tikzpicture}
    \caption{\hbox{Erfc}($z$) and its approximations. The green and orange circles are the two approximations used in \eqref{eq:Erfc_Aprox2}. The blue circles correspond to the approximation in \eqref{eq:erfcAprox_ICTON}. The red vertical line at $z=0$ separates the regions with positive and negative arguments of \hbox{Erfc}($z$).}
    \label{FigERFC}
\end{figure}
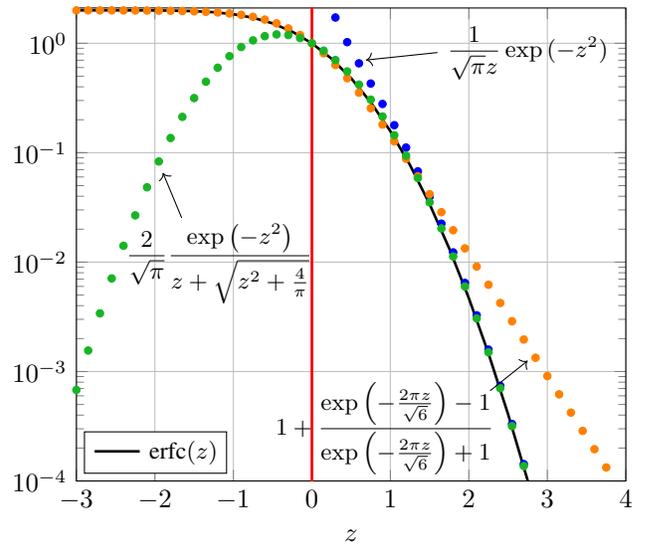

The second approximation is \cite[eq.~(6)]{erfc_Approx}
\begin{equation}
    \hbox{erfc}(z) \approx \frac{1}{z\sqrt{\pi}}\exp{(-z^2)}.
    \label{eq:erfcAprox_ICTON}
\end{equation}
The approximation in \eqref{eq:erfcAprox_ICTON} is also asymptotically tight (as \hbox{$z\rightarrow \infty$}). This approximation is shown in Fig.~\ref{FigERFC} with blue circles.
 
\subsection{Bit Error Rate: OOK}\label{BER_OOK}

For OOK signaling ($m=1$, $M=2$, and $\mathcal{X}=\{0,2P\}$), \eqref{eq:BER_OOK_conditioned} can be expressed as
\begin{equation}
    P_{b,2}=Q\left( \frac{\eta P h}{\sqrt{\sigma_n^2}} \right)=\frac{1}{2}\hbox{erfc}\biggl( \frac{\eta P h}{\sqrt{2\sigma_n^2}}\biggr).
    \label{Eq:CBER}
\end{equation}
We showed in \cite[eq.~(16)]{ICTON} that the average BER in \eqref{eq:averageBER_PDF} can be expressed as
\begin{IEEEeqnarray}{cl}
    \nonumber
&\overline{P}_{b,2}=\frac{\gamma^2}{\pi(h_g h_l \kappa)^{\gamma^2}}\exp{\left[ \gamma^2\sigma_R^2 \left(1+\gamma^2/2\right)\right]}\cdot\\
      &\int_{0}^{\infty}  h^{\gamma^2 -1} \left[\int_{uh}^{\infty} \hspace{-0.1cm}\exp{(-t^2)} \, dt \int_{v}^{\infty} \hspace{-0.1cm}\exp{(-t^2)}\, dt\right]\,\hspace{-0.15cm} dh,
        \label{eq:ExactOOKBER}
\end{IEEEeqnarray}
where $v$ is given in \eqref{psi}, and 
\begin{equation}
    u\triangleq \frac{\eta P}{\sqrt{2\sigma_n^2}}=\frac{A}{h\sqrt{8}}.
    \label{eq:u}
\end{equation}
The values of $u$ and $A=2P \eta h/ \sigma_n$ can be interpreted as variables that provide information about the quality of the transmission. 

From~\eqref{eq:ExactOOKBER}, we can see that this exact BER expression involves the integral of the product of two integrals, which is complex to evaluate.
In~\cite[eq.~(4)]{OECC}, we proposed the following BER approximation for OOK,  
\begin{align} 
        \nonumber
&\overline{P}_{b,2}\approx\frac{\gamma^2}{2\sqrt{\pi}(h_g h_l \kappa)^{\gamma^2}}\exp{\left[ \gamma^2\sigma_R^2 \left(1+\gamma^2/2\right)\right]}\cdot \hspace{1cm} \\
      &\Bigg[\int_{0}^{\hat{h}} 
      h^{\gamma^2 -1} \Biggl[1+ \dfrac{\exp{\left(\frac{-2\pi v}{\sqrt{6}}\right)}-1}{\exp{\left(\frac{-2\pi v}{\sqrt{6}}\right)}+1} \Biggr] 
      \frac{\exp{\left[-\left(uh\right)^2\right]}}{\left(uh\right)+\sqrt{\left(uh\right)^2+\frac{4}{\pi}}}\, dh \nonumber\\
      &+ \frac{2}{\sqrt{\pi}} \int_{\hat{h}}^{\infty} h^{\gamma^2 -1} \frac{\exp{\left(- v^2\right)} }{ v +\sqrt{ v^2 + \frac{4}{\pi}}} \frac{\exp{\left[-\left(uh\right)^2\right]}}{\left(uh\right)+\sqrt{\left(uh\right)^2+\frac{4}{\pi}}} \, dh \Biggr],
    \label{eq:BER_Approx2_UpperBound}
 \end{align}
where $\hat{h}\triangleq h_gh_l \kappa \exp{(-\mu)}$, $h_{l}$ is given in \eqref{eq:hl}, $h_{g}$ in \eqref{eq:hg} and $\mu$ in \eqref{eq:mu}. The expression in~\eqref{eq:BER_Approx2_UpperBound} is obtained by applying the approximation in \eqref{eq:Erfc_Aprox2} to the exact BER \eqref{eq:ExactOOKBER}.

Since~\eqref{eq:BER_Approx2_UpperBound} only involves two one-dimensional integrals, it is clear that it is easier to compute than \eqref{eq:ExactOOKBER}. An even simpler approximation can be obtained by using~\eqref{eq:erfcAprox_ICTON} in the exact BER \eqref{eq:ExactOOKBER}. This was indeed the approach taken in \cite{ICTON}, which leads to
\begin{IEEEeqnarray}{cc}
\nonumber
&\overline{P}_{b,2}\approx\frac{\gamma^2\sigma_R \sigma_n}{2\eta P\pi(h_g h_l \kappa)^{\gamma^2}} \exp{\left[ \gamma^2\sigma_R^2 \left(1+\gamma^2/2\right)\right]}\cdot\\
&\int_{\hat{h}}^{\infty} \frac{h^{\gamma^2 -2}}{\hbox{ln}(\frac{h}{h_gh_l\kappa})+\mu} \exp{\left[-\left(uh\right)^2\right]} \exp{\left(-v^2\right)}  \, dh.
\label{eq:BER_Aprox2_ICTON}
\end{IEEEeqnarray}
The expression in \eqref{eq:BER_Aprox2_ICTON} involves only a single one-dimensional integral. However, as we will see in Sec.~\ref{sec:results},~\eqref{eq:BER_Approx2_UpperBound} is significantly more accurate than~\eqref{eq:BER_Aprox2_ICTON}, creating a trade-off between complexity and accuracy. 

\subsection{Bit Error Rate: $M$-PAM}
The BER expressions~\eqref{eq:ExactOOKBER},~\eqref{eq:BER_Approx2_UpperBound}, and~\eqref{eq:BER_Aprox2_ICTON} are only valid for OOK. In this paper, we are interested in computing the BER for $M$-PAM. Instead of computing~\eqref{eq:averageBER_PDF}, in this section we will show that because of \eqref{eq:BER_MPAM}, computing $\overline{P}_{s,M}$ is enough. To demonstrate the accuracy of \eqref{eq:BER_MPAM}, we analyze the error probability of $8$-PAM and $16$-PAM. Deriving the BER expression is not done for $M>16$ as it becomes too tedious.

For simplicity, instead of considering the average BER expression $\overline{P}_{b,M}$, we present here the analysis based on the BER conditioned on \mbox{$H=h$}. The conditional BER in \eqref{eq:BER_OOK_conditioned} for $8$-PAM and $16$-PAM are given by
\begin{align}
\nonumber
    P_{b,8} =& \frac{1}{12}\Bigg[ 7Q\left( \frac{A}{14} \right) +6 Q \left( \frac{3A}{14} \right)- \\ 
    &Q\left( \frac{5A}{14} \right) +Q \left( \frac{9A}{14} \right) - Q \left( \frac{13A}{14} \right)\Bigg],
    \label{TrueBER_8PAM}
\end{align}
and 
\begin{align}
\nonumber
    &P_{b,16} = \frac{1}{32}\Bigg[ 15Q\left( \frac{A}{30} \right) +14 Q \left( \frac{3A}{30} \right) -Q\left( \frac{5A}{30} \right)+\\\nonumber
    &\sum_{k=0}^2 \frac{5}{(-1)^k}Q\left( \frac{(4k+9)A}{30}\right)+
    \frac{4}{(-1)^k} Q\left( \frac{(4k+11)A}{30}\right)+\\ &\sum_{k=0}^1 (-1)^k Q\left( \frac{(4k+25)A}{30}\right) -3Q\left( \frac{21A}{30} \right)-2Q\left( \frac{23A}{30} \right)\Bigg],
    \label{TrueBER_16PAM}
\end{align}
resp. For $8$-PAM and $16$-PAM, the approximation~\eqref{eq:BER_MPAM} yields
\begin{equation}
    P_{b,8}\approx\hat{P}_{b,8}=\frac{P_{s,8}}{3}=\frac{1}{12}\left[ 7Q\left( \frac{A}{30} \right)\right],
    \label{AproxBER_8PAM}
\end{equation}
and
\begin{equation}\label{AproxBER_16PAM}
 P_{b,16}\approx\hat{P}_{b,16}=\frac{P_{s,16}}{4}=\frac{1}{32}\left[ 5Q\left( \frac{A}{30} \right)\right]. 
\end{equation}
In \eqref{AproxBER_8PAM} and \eqref{AproxBER_16PAM}, $\hat{P}_{b,M}$ denotes the approximated conditional BER. Furthermore, in \eqref{AproxBER_8PAM} and \eqref{AproxBER_16PAM}, we use the SER expression for $M$-PAM given in \eqref{Eq:CSER_PAM}.

In the high SNR regime, which corresponds to the case $A \to \infty$ (see \eqref{eq:u}), we have that
\begin{equation}\label{eq:lim_HighSNR}
    \lim_{A \to \infty} \frac{P_{b,8}}{\hat{P}_{b,8}}=\lim_{A \to \infty} \frac{P_{b,16}}{\hat{P}_{b,16}}=1,
\end{equation} 
which follows from considering the dominant Q-functions in \eqref{TrueBER_8PAM} and \eqref{TrueBER_16PAM}.
The result in \eqref{eq:lim_HighSNR} shows that the approximations $\hat{P}_{b,M}$ are expected to be accurate, especially at high SNR.

In Fig. \ref{FigBER_MPAM}, we compare the exact BER in \eqref{eq:averageBER_PDF} and the BER approximation \eqref{eq:BER_MPAM} for $8$-PAM and $16$-PAM. 
More specifically, we use the conditional BERs in \eqref{TrueBER_8PAM} and \eqref{TrueBER_16PAM} to compute the exact average BER $\overline{P}_{b,M}$ for $8$-PAM and $16$-PAM, resp. Additionally, we use the conditional BER approximations in \eqref{AproxBER_8PAM} and~\eqref{AproxBER_16PAM}, to compute the approximated average BERs for $8$-PAM and $16$-PAM signaling, resp.
Blue and gray curves represent $8$-PAM and $16$-PAM, resp. Solid curves correspond to the exact BER while dotted curves correspond to the approximate average BER. Out of the nine operating points shown in Fig.~\ref{FigPDF}, three representative cases have been selected for further analysis.
These are: (i) $\sigma_s = 0.35$~m and $\sigma_R^2 = 0.1$, (ii) $\sigma_s = 0.25$~m and $\sigma_R^2 = 0.5$, and (iii) $\sigma_s = 0.2$~m and $\sigma_R^2 = 0.9$. These operating points will be used throughout the remainder of the paper.
From this figure, we can see that the approximation is in close agreement with the exact BER as the transmitted power increases.

The rightmost curves in Fig. \ref{FigBER_MPAM} correspond to the most severe turbulence ($\sigma_R^2=0.9$). This case is associated with the operating point that exhibits the smallest mean channel gain $\mathbb{E}[H]$ (see Table~\ref{tab:Values_E_h}). Therefore, for the same transmitted power, this operating point has the lowest optical SNR (see \eqref{SNR_op_av}). The approximation in \eqref{eq:BER_MPAM} is asymptotically tight for high SNR, as shown in \eqref{eq:lim_HighSNR}, which explains why, for this operating point, the solid and dotted curves begin to converge at higher transmitted power values.  

\begin{figure}[t]
        \centering
            \begin{tikzpicture}[scale=1]    

\definecolor{my_pink}{rgb}{1,0.07,0.65}
\definecolor{my_g}{rgb}{0.09,0.71,0.14}
\definecolor{FEC}{rgb}{0.5, 0.0, 0.13}
 \definecolor{green2}{rgb}{0.0, 0.5, 0.5}
\definecolor{fuchsia}{rgb}{0.6,0.4,0.8}
\definecolor{my_purple}{rgb}{0.5,0,0.5}

\begin{semilogyaxis}[ 
    width=0.49\textwidth,
    height=3.5in,   
    xmin=0, xmax=30,
    ymin=1e-6,
    ymax=6e-1,  
    xlabel={Transmitted Power, $P$~[dBm]}, 
    ylabel shift=-4ex,
    xlabel shift=-4ex,
    font=\small,
    ylabel={Average BER,~$\overline{P}_{b,M}$}, 
    xtick={0,5,...,30},    
    grid=major,
    grid style={dashed,lightgray!75},
]

\addplot [color=black,solid,line width=1pt,mark options={solid}]coordinates {
         (10,1) (10,2)};\label{Orig_var1_jit12_black}
\addplot [color=black,dashed,line width=1pt,mark options={solid}]coordinates {
         (10,1) (10,2)};\label{Dashed_black}
\addplot [color=black,dotted,line width=1.2pt,mark options={solid}]coordinates {(10,1) (10,2)};\label{dotted_black}

  

\addplot [color=gray,solid,line width=1pt,mark options={solid}]file{./txtData/16PAM_BER/True_varR09_jit2.txt};\label{C16PAM}

\addplot [ color=gray,dotted,line width=1.2pt,mark options={solid}] file {./txtData/16PAM_BER/Aprox_varR09_jit2.txt};

    
    \addplot [color=blue!70,solid,line width=1pt,mark options={solid}]file{./txtData/8PAM_BER/True_varR09_jit2.txt};\label{C8PAM}
    
    \addplot [ color=blue!70,dotted,line width=1.2pt,mark options={solid}] file {./txtData/8PAM_BER/Aprox_varR09_jit2.txt};


\addplot [color=gray,solid,line width=1pt,mark options={solid}]file{./txtData/16PAM_BER/True_varR01_jit35.txt};

\addplot [ color=gray,dotted,line width=1.2pt,mark options={solid}] file{./txtData/16PAM_BER/Aprox_varR01_jit35.txt};

\addplot [color=blue!70,solid,line width=1pt,mark options={solid}]file{./txtData/8PAM_BER/True_varR01_jit35.txt};

\addplot [ color=blue!70,dotted,line width=1.2pt,mark options={solid}] file{./txtData/8PAM_BER/Aprox_varR01_jit35.txt};


\addplot [color=gray,solid,line width=1pt,mark options={solid}] file {./txtData/16PAM_BER/True_varR05_jit25.txt};

\addplot [color=gray,dotted,line width=1pt,mark options={solid}] file{./txtData/16PAM_BER/Aprox_varR05_jit25.txt};

\addplot [color=blue!70,solid,line width=1pt,mark options={solid}] file {./txtData/8PAM_BER/True_varR05_jit25.txt};

\addplot [color=blue!70,dotted,line width=1pt,mark options={solid}] file{./txtData/8PAM_BER/Aprox_varR05_jit25.txt};

\end{semilogyaxis} 

\node [coordinate](input) {};

\node [draw,color=my_pink, fill=white,anchor= south west,font=\scriptsize,text=black] at (1.2,3.4 em) {\shortstack[l]{ 
$\sigma_s=0.35$ m \\
$\sigma_R^2=0.1$}};

\node [draw,color=my_g, fill=white,anchor= south 
west,font=\scriptsize,text=black] at (5.7,8.8em) {\shortstack[l]{ 
$\sigma_s=0.2$ m \\
$\sigma_R^2=0.9$}};

\node [draw,color=orange, fill=white,anchor= south 
west,font=\scriptsize,text=black] at (1.4,0.4 em) {\shortstack[l]{ 
$\sigma_s=0.25$ m \\
$\sigma_R^2=0.5$}};

\draw[stealth-,solid, black,line width=0.7pt,rounded corners] ($(input.north)+(4.65,2.8em)$) -- ($(input.north)+(3.25,1.3em)$);

\draw[stealth-,solid, black,line width=0.7pt,rounded corners]($(input.north)+(3.8,16.5em)$) -- ($(input.north)+(1.5,15.5em)$);
\node at ($(input.north) + (4.5, 17 em)$) [scale=1.2, font=\scriptsize] {Increasing};
\node at ($(input.north) + (4.5, 16.3 em)$) [scale=1.2, font=\scriptsize] {Turbulence};


\node[ellipse, color=my_pink,line width=0.7pt, draw,minimum width = 1.1cm, 
	minimum height = 0.45cm, rotate=45] (e) at (3.45,1.8) {};

\node[ellipse, color=my_g,line width=0.7pt, draw,minimum width = 1cm, 
	minimum height = 0.45cm, rotate=45] (e) at (5.7,2.6) {};

\node[ellipse, color=orange,line width=0.7pt, draw,minimum width = 1cm, 
	minimum height = 0.45cm, rotate=45] (e) at (5.1,1.35) {};

\node [draw,fill=white,anchor= south west,font=\scriptsize] at (5.28,6.55) {\shortstack[l]{ 
\ref{C8PAM} $M=8$\\
\ref{C16PAM} $M=16$}};

\node [draw,fill=white,anchor= south west,font=\scriptsize] at (5.82,5.7) {\shortstack[l]{ 
\ref{Orig_var1_jit12_black} \eqref{eq:averageBER_PDF} \\
\ref{dotted_black} \eqref{eq:BER_MPAM}}};

\end{tikzpicture}
        \vspace{-2ex}
    \caption{Average BER in \eqref{eq:averageBER_PDF} (exact) and \eqref{eq:BER_MPAM} (approximation) versus transmit power for $8$-PAM and $16$-PAM, and three different cases for ($\sigma_s$, $\sigma_R^2$): ($0.2,0.9$), ($0.25,0.5$), and ($0.35,0.1$). Table~\ref{tab:System_params} summarizes the simulation parameters considered (Sec.~\ref{sec:results}, ahead).}
    \label{FigBER_MPAM}
\end{figure}
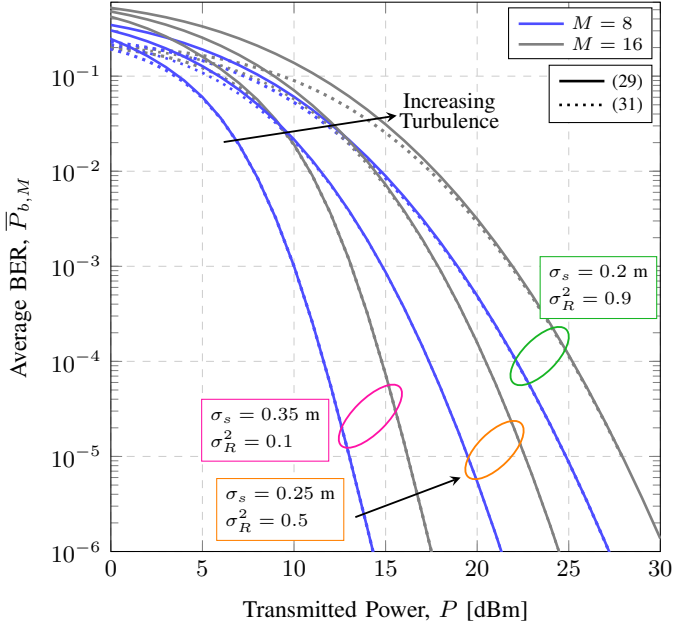


\subsection{Symbol Error Rate: $M$-PAM}
\label{SER_sub}

Here, we derive the exact SER expression and develop novel approximations for $M$-PAM for the composite channel model in \eqref{eq:fh_LN_sim}. In the following theorem, we present the exact average SER. To the best of our knowledge, this expression has not been derived in the literature for $M$-PAM over a composite fading channel like the one considered in this work. 

\begin{theorem}[\textit{Exact SER}]
\label{th:SER_exact}
The exact average SER for $M$-PAM over a composite fading channel characterized by the PDF in \eqref{eq:fh_LN_sim} is given by
\begin{align}
  \nonumber
     &\overline{P}_{s,M}=\frac{2(M-1)\gamma^2}{\pi M(h_g h_l \kappa)^{\gamma^2}}\exp{\left[ \gamma^2\sigma_R^2 \left(1+\gamma^2/2\right)\right]}\cdot \hspace{4cm}\\ &\int_{0}^{\infty}  h^{\gamma^2 -1} \left[ \int_{u h/(M-1)}^{\infty} \exp{\left(-t^2 \right)} dt  \int_{v}^{\infty}  \exp{\left(-t^2 \right)} dt \right] dh, \hspace{2.5cm}
     \label{eq:SER_exact_PAM}
\end{align}
where $v$ and $u$ are as given in \eqref{psi} and \eqref{eq:u}, resp.
\end{theorem}
\begin{IEEEproof}
   By using \eqref{eq:fh_LN_sim}, \eqref{eq:erfc}, and \eqref{Eq:CSER_PAM} in \eqref{eq:averageSER_PDF}.
\end{IEEEproof}

Similar to~\eqref{eq:ExactOOKBER}, the exact SER expression given in Theorem~\ref{th:SER_exact} again involves the integral of the product of two integrals, which is complex to evaluate.
In order to simplify the computation of \eqref{eq:SER_exact_PAM}, the following theorem gives an SER approximation that only involves two one-dimensional integrals.  

\begin{theorem}[\textit{Approximated SER for $M$-PAM}]
The average SER in \eqref{eq:SER_exact_PAM} can be approximated as
     \begin{align}  
\nonumber
      &\overline{P}_{s,M}\approx\frac{(M-1)\gamma^2}{M\sqrt{\pi}(h_g h_l \kappa)^{\gamma^2}}\exp{\left[ \gamma^2\sigma_R^2 \left(1+\gamma^2/2\right)\right]}  \Bigg[\int_{0}^{\hat{h}} \hspace{-0.2cm} 
      h^{\gamma^2 -1}  \cdot \hspace{2.5cm}\\ &\nonumber \left[1+\dfrac{\exp{\left(\frac{-2\pi v }{\sqrt{6}}\right)}-1}{\exp{\left(\frac{-2\pi v }{\sqrt{6}}\right)}+1} \right] 
      \frac{\exp{\left[-\left(\frac{uh}{M-1}\right)^2\right]}}{\left(\frac{uh}{M-1}\right)+\sqrt{\left(\frac{uh}{M-1}\right)^2+\frac{4}{\pi}}}\, dh \hspace{2cm} \\ &+\hspace{-0.1cm} \frac{2}{\sqrt{\pi}} \hspace{-0.12cm} \int_{\hat{h}}^{\infty}\hspace{-0.2cm} h^{\gamma^2 -1} \frac{\exp{(- v^2)} }{ v+\sqrt{ v^2 + \frac{4}{\pi}}} \frac{\exp{\left[-\left(\frac{uh}{M-1}\right)^2\right]}}{\left(\frac{uh}{M-1}\right)+\sqrt{\left(\frac{uh}{M-1}\right)^2+\frac{4}{\pi}}} \, dh\Biggr]. \hspace{2ex}
     \label{eq:SER_Approx2_UpperBound}
 \end{align}
 \label{Theorem:SER_Aprox}
\end{theorem}
\begin{IEEEproof}
    By using the approximation \eqref{eq:Erfc_Aprox2} in \eqref{eq:SER_exact_PAM}.
\end{IEEEproof}

Since \eqref{eq:SER_Approx2_UpperBound} only involves two one-dimensional integrals, it is easier to compute than \eqref{eq:SER_exact_PAM}. In Sec.~\ref{sec:results}, we will show how accurate this approximation is. 

\subsection{Symbol Error Rate: Dense PAM Constellations}

An $M$-PAM constellation is said to be \textit{dense} if the number of constellation points $M$ is much larger than 2 ($M\gg2$). Dense $M$-PAM constellations have tightly spaced amplitude levels within the values $0$ and $2P$. Dense $M$-PAM constellations are useful for analyzing certain asymptotic behaviors, including for example obtaining the so-called ultimate shaping gain \cite{AnalysisDensePAM}. In this section we present some results on dense $M$-PAM constellations for the composite channel model in \eqref{eq:fh_LN_sim}.

Building on Theorem~\ref{Theorem:SER_Aprox}, the next corollary presents an approximation for the average SER for dense $M$-PAM constellations. Its proof follows directly by replacing \hbox{$M-1$} by $M$ in \eqref{eq:SER_Approx2_UpperBound}.

\begin{corollary}[\textit{Approximated SER for Dense $M$-PAM}]
The average SER for dense $M$-PAM constellations can be approximated as
    \begin{align}  
\nonumber
      &\overline{P}_{s,M}\approx\frac{\gamma^2}{\sqrt{\pi}(h_g h_l \kappa)^{\gamma^2}}\exp{\left[ \gamma^2\sigma_R^2 \left(1+\gamma^2/2\right)\right]}  \Bigg[\int_{0}^{\hat{h}} \hspace{-0.2cm} 
      h^{\gamma^2 -1}  \cdot \hspace{2.5cm}\\ \nonumber &\left[1+\dfrac{\exp{\left(\frac{-2\pi v }{\sqrt{6}}\right)}-1}{\exp{\left(\frac{-2\pi v }{\sqrt{6}}\right)}+1} \right] 
      \frac{\exp{\left[-\left(\frac{uh}{M}\right)^2\right]}}{\left(\frac{uh}{M}\right)+\sqrt{\left(\frac{uh}{M}\right)^2+\frac{4}{\pi}}}\, dh +\\ 
      & \frac{2}{\sqrt{\pi}} \int_{\hat{h}}^{\infty} h^{\gamma^2 -1} \frac{\exp{(- v^2)} }{ v+\sqrt{ v^2 + \frac{4}{\pi}}} \frac{\exp{\left[-\left(\frac{uh}{M}\right)^2\right]}}{\left(\frac{uh}{M}\right)+\sqrt{\left(\frac{uh}{M}\right)^2+\frac{4}{\pi}}} \, dh\Biggr]. 
     \label{eq:SER_DensePAM}
 \end{align}
 \label{Theorem:SER_densePAM}
\end{corollary}

In a communication system using $M$-PAM, one way to improve the spectral efficiency by $1$ bit/symbol is to double the modulation order $M$. However, to maintain the same SER, the system would require a higher transmit power. 
This increase in power has significant implications for the design of FSO systems, as available power is limited due to eye safety regulations. Consequently, this factor must be carefully considered when designing FSO links.
The following corollary quantifies the additional power required for dense $M$-PAM constellations when
the spectral efficiency is increased by $1$ bit/symbol while maintaining the same SER.

\begin{corollary}[\textit{Doubling Power for Dense $M$-PAM}]
\label{Corollary1}
For dense $M$-PAM constellations and any given SER, doubling the transmitted power allows an increase in spectral efficiency of 1~bit/symbol.
\end{corollary}
\begin{IEEEproof}
The SER expression in \eqref{eq:SER_DensePAM} depends on the modulation order $M$ only through the term $\left({uh}/{M}\right)^2$, which appears in both integrals in \eqref{eq:SER_DensePAM}, as part of the common term
\begin{equation}\label{proof.integrand}
\frac{\exp{\left[-\left(\frac{uh}{M}\right)^2\right]}}{\left(\frac{uh}{M}\right)+\sqrt{\left(\frac{uh}{M}\right)^2+\frac{4}{\pi}}}.
\end{equation}
The proof is completed by noting that $u$ in \eqref{eq:u} depends linearly on $P$. Therefore, doubling the power while also doubling $M$ (i.e., increasing spectral efficiency by 1~bit/symbol) will result in identical integrands in \eqref{eq:SER_DensePAM} (see \eqref{proof.integrand}), and thus identical SERs. 
\label{Proof1}
\end{IEEEproof}

The last result in this section is a further simplification of the SER for dense PAM constellations, which offers negligible loss in accuracy in the high transmit power regime. The following theorem presents this result.

\begin{theorem}[\textit{Approximated SER for Dense $M$-PAM and High Transmit Power}]
The average SER for dense $M$-PAM constellations  at high transmit powers can be approximated as
\begin{align}  
\nonumber
     &\overline{P}_{s,M}\approx\frac{M\gamma^2}{2u\sqrt{\pi}(h_g h_l \kappa)^{\gamma^2}}\exp{\left[ \gamma^2\sigma_R^2 \left(1+\gamma^2/2\right)\right]}\cdot\\   \nonumber &\Bigg[\int_{0}^{\hat{h}}
      h^{\gamma^2-2}\left[1+\dfrac{\exp{\left(\frac{-2\pi v }{\sqrt{6}}\right)}-1}{\exp{\left(\frac{-2\pi v }{\sqrt{6}}\right)}+1} \right] 
      \exp{\left[-\left(\frac{uh}{M}\right)^2\right]}\, dh \\ &+ \frac{2}{\sqrt{\pi}} \int_{\hat{h}}^{\infty}h^{\gamma^2-2} \frac{\exp{(- v^2)} }{ v+\sqrt{ v^2 + \frac{4}{\pi}}}\exp{\left[-\left(\frac{uh}{M}\right)^2\right]} \, dh\Biggr]. 
     \label{eq:SER_DenseMPAM_No4Pi}
 \end{align}
 \label{Theorem:SER_NoPi}
\end{theorem}
\begin{proof}
    By using $\left(uh/M\right)^2+ 4/\pi \approx \left(uh/M\right)^2$ in \eqref{eq:SER_DensePAM}.
\end{proof}

In Sec.~\ref{sec:results}, we discuss the accuracy of \eqref{eq:SER_Approx2_UpperBound} and the simple new approximation proposed in \eqref{eq:SER_DenseMPAM_No4Pi} for dense $M$-PAM constellations.

\vspace{-1ex}
\section{Numerical Results}
\label{sec:results}
In this section, we compare the accuracy of the approximate BER and SER expressions in \eqref{eq:BER_MPAM}, \eqref{eq:BER_Approx2_UpperBound}, \eqref{eq:BER_Aprox2_ICTON}, \eqref{eq:SER_Approx2_UpperBound}, and \eqref{eq:SER_DenseMPAM_No4Pi}, presented in Sec.~\ref{Sec:SER} for realistic link conditions. 
Additionally, we 
assess the impact of increasing the modulation order on the required transmitted power.

Table~\ref{tab:System_params} summarizes the parameters considered in our study.
These parameters are mostly based on~\cite{Farid2007,Korevaar03,Boluda2017}.
In particular, we consider a point-to-point $z=3$~km link where the transmitter sends a light beam with a divergence angle $\theta=1.32$ mrad operating at a wavelength of $\lambda=1550$~nm. The receiver is equipped with an aperture radius~$a=5$~cm with a responsivity $\beta=0.5$ A/W.

In terms of pointing errors and turbulence, the same three operating points considered in Fig.~\ref{FigBER_MPAM} have been considered.
The first ($\sigma_s=0.35$~m and $\sigma_R^2=0.1$), the second ($\sigma_s=0.25$~m and $\sigma_R^2=0.5$), and third ($\sigma_s=0.2$~m and $\sigma_R^2=0.9$) cases are shown in the figures with pink, orange, and green, resp.
Finally, we consider the attenuation coefficient \mbox{$\sigma(\lambda)=0.2208$~dB/km} corresponding to clear air weather condition, which results in $h_l=0.516$ and the attenuation due to the geometric spread $h_g=1.3\cdot10^{-3}$.

\subsection{Bit Error Rate}

In Fig.~\ref{FigBER_OOK}, the average BER for OOK signaling is plotted as a function of the average transmit power $P$.
Solid, dotted, and dashed curves correspond to the exact average BER~\eqref{eq:ExactOOKBER}, the approximate average BER in~\eqref{eq:BER_Approx2_UpperBound} (from \cite{OECC}), and the approximate average BER in~\eqref{eq:BER_Aprox2_ICTON} (from \cite{ICTON}), resp.
BER estimates obtained via Monte Carlo simulations are also shown with circles. 
Moreover, a HD FEC threshold of $3.84\times10^{-3}$ for a coding rate 0.937~\cite[Table I]{FEC} is also shown. The pre-FEC BER must be lower than this HD FEC threshold in order to achieve a post-FEC BER lower than $10^{-15}$. 
We define an accuracy parameter $\Delta$ as the difference in transmit power $P$ at which the exact BER~\eqref{eq:ExactOOKBER} and our proposed approximation~\eqref{eq:BER_Approx2_UpperBound} reach the HD FEC BER limit.

\begin{table}[t]
\begin{center}
\caption{FSO System Configuration} \label{tab:System_params}  
\small
\renewcommand{\arraystretch}{1.2}
\begin{tabular}{cc}
\hline
 \multicolumn{2}{c}{{\textbf{FSO Parameters}}} \\
            \hline
\textbf{Parameter} & \textbf{Value}  \\
\hline
Wavelength ($\lambda$)  & 1550 nm  \\
Link distance ($z$) & 3 km \\
Noise  standard deviation ($\sigma_n$)& $10^{-7}$~A  \\
Electro-optical conversion factor ($\alpha$) & 1~W/A\\
Responsivity ($\beta$) & 0.5~A/W  \\
 Tx divergence angle ($\theta$) & $1.32$ mrad  \\
  Beam Waist ($\omega_z$) &   $1.98$ m\\
Receiver radius ($a$) & $5$ cm  \\
Geometric spread attenuation ($h_g$)& $1.3\cdot10^{-3}$\\
\hline
  \multicolumn{2}{c}{{\textbf{Turbulence, Pointing Error and Weather Parameters}}} \\
             \hline
         \textbf{Parameter} & \textbf{Value}  \\
            \hline
             Attenuation coefficient ($\sigma(\lambda)$)   & $0.2208$~dB/km\\ 
             Atmospheric Attenuation ($h_l$)   & $0.516$\\
              Rytov variance ($\sigma_R^2$) &   $\leq 1$\\ 
            Jitter angle ($\theta_{s}$) & $0.067-0.116$ mrad \\
            \hline
\end{tabular}
\end{center}
\end{table}

Fig.~\ref{FigBER_OOK} shows that the best BER performance is achieved in the first case (pink curves). 
Although this case corresponds to the most severe pointing errors, the dominant effect for the operating points under consideration is the atmospheric turbulence which is the weakest for this case.
For the first case (high pointing errors and very weak turbulence, pink curves), the BER approximation~$\eqref{eq:BER_Aprox2_ICTON}$ proposed in~\cite{ICTON} is $1.13$~dB away from the exact BER~\eqref{eq:ExactOOKBER} at the HD FEC limit.
For the second (medium pointing errors and weak turbulence, orange curves) and third cases (the lowest pointing errors and the most severe weak turbulence, green curves), the approximation~\eqref{eq:BER_Aprox2_ICTON} still results in imprecise estimates, off by $0.44$~dB and $0.69$~dB at the HD-FEC threshold, resp.
The new approximation~\eqref{eq:BER_Approx2_UpperBound} (initially  proposed in~\cite{OECC}) however, has a gap of only \hbox{$\Delta=0.20$~dB} for the first case, while for the second and third cases, the gap $\Delta$ is below $0.09$~dB. 
Additionally, we see in Fig.~\ref{FigBER_OOK} that for the BER values of interest (i.e., BERs between $10^{-2}-10^{-5}$), operating in the high pointing error regime, the slope of the BER approximation in~\eqref{eq:BER_Aprox2_ICTON} (pink dashed curve) differs from the slope of the Monte Carlo results. 
Our last proposed approximation~\eqref{eq:BER_Approx2_UpperBound}, however, seems to provide BER curves with correct slopes.

\begin{figure}[t]
        \centering
            \begin{tikzpicture}[scale=1]    

\definecolor{my_pink}{rgb}{1,0.07,0.65}
\definecolor{my_g}{rgb}{0.09,0.71,0.14}
\definecolor{FEC}{rgb}{0.5, 0.0, 0.13}

\definecolor{fuchsia}{rgb}{0.6,0.4,0.8}

\begin{semilogyaxis}[ 
    width=0.49\textwidth,
    height=3.5in,   
    xmin=-4, xmax=20,
    ymin=1e-6, ymax=3e-1,  
    xlabel={Transmitted Power, $P$~[dBm]}, 
    ylabel shift=-4ex,
    xlabel shift=-4ex,
    font=\small,
    ylabel={Average BER, $\overline{P}_{b,2}$}, 
    xtick={-4,0,...,20},    
    grid=major,
    grid style={dashed,lightgray!75},
]

\addplot [color=black,solid,line width=1pt,mark options={solid}]coordinates {
         (10,1) (10,2)};\label{Orig_var1_jit12_black}
\addplot [color=black,dashed,line width=1pt,mark options={solid}]coordinates {
         (10,1) (10,2)};\label{Dashed_black}
\addplot [color=black,dotted,line width=1pt,mark options={solid}]coordinates {(10,1) (10,2)};\label{dotted_black}
 \addplot [only marks, color=black, mark=*, mark options={solid,scale=0.8, fill=white}]coordinates {(0.35,1)};\label{MC}

  \addplot [name path=lim, color=blue, line width=1pt]
   coordinates {(-10,3.84e-3) (21,3.84e-3)};\label{BERlimit}
  

\addplot [color=my_g,solid,line width=1pt,mark options={solid}]file{./txtData/SER_PAM/varR09_jit2/OOK_SER_orig.txt};

\addplot [ color=my_g,dashed,line width=1pt,mark options={solid}] file {./txtData/BER_OOK_ICTON/Aprox2_varR09_jit2.txt};

\addplot [color=my_g,dotted,line width=1pt,mark options={solid}] file{./txtData/SER_PAM/UpperBound/varR09_jit2/OOK_aprox2.txt};

\addplot [only marks, color=my_g, mark=*, mark options={solid,scale=0.8, fill=white}]file{./txtData/SER_PAM/varR09_jit2/OOK_SER_MC.txt};\


\addplot [color=my_pink,solid,line width=1pt,mark options={solid}]file{./txtData/SER_PAM/varR01_jit35/OOK_SER_orig.txt};\label{Orig_var01_jit35}

\addplot [ color=my_pink,dotted,line width=1pt,mark options={solid}] file{./txtData/SER_PAM/UpperBound/varR01_jit35/OOK_aprox2.txt};
    
\addplot [color=my_pink,dashed,line width=1pt,mark options={solid}] file{./txtData/BER_OOK_ICTON/Aprox2_varR01_jit35.txt};


\addplot [only marks, color=my_pink, mark=*, mark options={solid,fill=white,scale=0.8}]file{./txtData/SER_PAM/varR01_jit35/OOK_MC.txt};\label{MC_orig_var01_jit35}


\addplot [color=orange,solid,line width=1pt,mark options={solid}] file {./txtData/SER_PAM/varR05_jit25/OOK_SER_Orig.txt};\label{Orig_var05_jit25}

\addplot [color=orange,dotted,line width=1pt,mark options={solid}] file{./txtData/SER_PAM/UpperBound/varR05_jit25/OOK_aprox2.txt};
    
\addplot [color=orange,dashed,line width=1pt,mark options={solid}] file{./txtData/BER_OOK_ICTON/Aprox2_varR05_jit25.txt};

\addplot [only marks, color=orange, mark=*, mark options={solid,scale=0.8,fill=white}]file{./txtData/SER_PAM/varR05_jit25/OOK_MC.txt};\label{MC_orig_var05_jit25}

\end{semilogyaxis} 

\node [draw,fill=white,anchor= south west,font=\scriptsize] at (4.8,5.6) {\shortstack[l]{ 
\ref{Orig_var1_jit12_black} \eqref{eq:ExactOOKBER}: Exact \\
 \ref{dotted_black} \eqref{eq:BER_Approx2_UpperBound}: \cite{OECC}\\
\ref{Dashed_black} \eqref{eq:BER_Aprox2_ICTON}: \cite{ICTON}\\
\hspace{0.17cm} \ref{MC} \hspace{0.2cm}  Monte Carlo\\
  \ref{BERlimit} HD-FEC BER}};

\node [coordinate](input) {};

\draw[-,dashed, black,line width=0.5pt] ($(input.north)+(1.5,13.7em)$) -- ($(input.north)+(1.5,14.7em)$);
\draw[-,dashed, black,line width=0.5pt] ($(input.north)+(1.82,13.7em)$) -- ($(input.north)+(1.82,14.7em)$);
\node at ($(input.north) + (1.9, 15 em)$) [ scale=1, font=\scriptsize] {$1.13$~dB};

\draw[-,dashed, black,line width=0.5pt] ($(input.north)+(4,13.6em)$) -- ($(input.north)+(4,14.8em)$);
\draw[-,dashed, black,line width=0.5pt] ($(input.north)+(4.2,13.6em)$) -- ($(input.north)+(4.2,14.8em)$);
 \node at ($(input.north) + (4.2, 15em)$) [ scale=1, font=\scriptsize] {$0.69$~dB};

\draw[-,dashed, black,line width=0.5pt] ($(input.north)+(2.83,13.7em)$) -- ($(input.north)+(2.83,14.7em)$);
\draw[-,dashed, black,line width=0.5pt] ($(input.north)+(2.95,13.7em)$) -- ($(input.north)+(2.95,14.7em)$);
 \node at ($(input.north) + (3.05, 15 em)$) [scale=1, font=\scriptsize] {$0.44$ dB};

\draw[-,dashed, black,line width=0.5pt] ($(input.north)+(1.5,13.7em)$) -- ($(input.north)+(1.5,12em)$);
\draw[-,dashed, black,line width=0.5pt] ($(input.north)+(1.55,13.7em)$) -- ($(input.north)+(1.55,12em)$);
\node at ($(input.north) + (1.3, 11.6 em)$) [ scale=1, font=\scriptsize] {$\Delta\approx0.20$~dB};

\draw[-,dashed, black,line width=0.5pt] ($(input.north)+(4,13.7em)$) -- ($(input.north)+(4,12em)$);
\node at ($(input.north) + (4.2, 11.6em)$) [scale=1, font=\scriptsize] {$\Delta\approx0.09$~dB};

\draw[-,dashed, black,line width=0.5pt] ($(input.north)+(2.83,13.7em)$) -- ($(input.north)+(2.83,11.7em)$);
\node at ($(input.north) + (2.75, 11em)$) [scale=1, font=\scriptsize] {$\Delta\approx0.07$ dB};

\draw[stealth-,solid, black,line width=0.7pt,rounded corners]($(input.north)+(2.9,18em)$) -- ($(input.north)+(0.8,16.5em)$);
\node at ($(input.north) + (3.6, 18.3 em)$) [scale=1.2, font=\scriptsize] {Increasing};
\node at ($(input.north) + (3.6, 17.7 em)$) [scale=1.2, font=\scriptsize] {Turbulence};


\node[ellipse,line width=0.7pt, draw,minimum width = 0.6cm, 
	minimum height = 0.05cm, rotate=45] (e) at (2.7,1.1) {};
\node [draw,color=my_pink, fill=white,anchor= south west,font=\scriptsize,text=black] at (0.7,0.2 em) {\shortstack[l]{ 
$\sigma_s=0.35$ m \\
$\sigma_R^2=0.1$}};
 
\node[ellipse,line width=0.7pt, draw,minimum width = 0.6cm, 
	minimum height = 0.05cm,rotate=45] (e) at (6.4,1.3) {};
\node [draw,color=my_g, fill=white,anchor= south west,font=\scriptsize,text=black] at (5.25,0.2 em) {\shortstack[l]{ 
$\sigma_s=0.2$ m \\
$\sigma_R^2=0.9$}};

\node[ellipse,line width=0.7pt, draw,minimum width = 0.6cm, 
	minimum height = 0.05cm,rotate=45] (e) at (4.8,1.3) {};
\node [draw,color=orange, fill=white,anchor= south west,font=\scriptsize,text=black] at (3.15,0.2 em) {\shortstack[l]{ 
$\sigma_s=0.25$~m \\
$\sigma_R^2=0.5$}};
 
\end{tikzpicture}
        \vspace{-2ex}
    \caption{Average BER versus transmit power considering OOK, clear air, $z=3$~km, and three different operating points for ($\sigma_s$, $\sigma_R^2$): ($0.2,0.9$), ($0.25,0.5$), and ($0.35,0.1$).}
    \label{FigBER_OOK}
\end{figure}
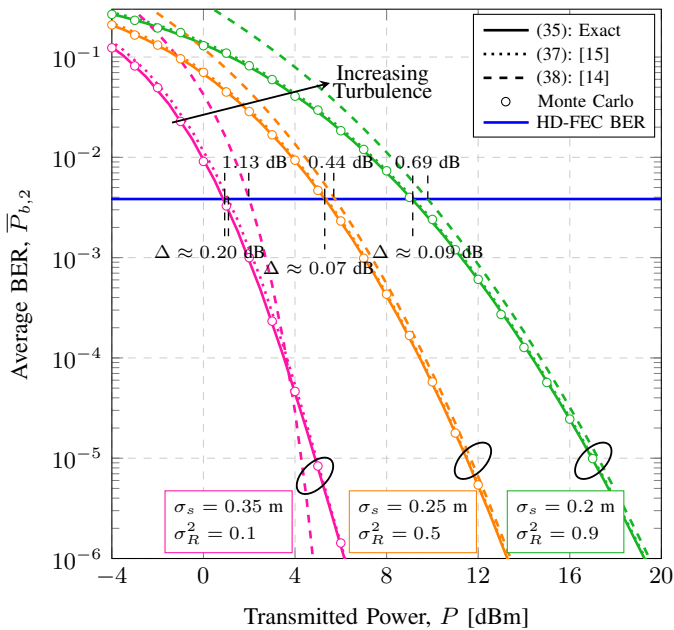


\begin{figure}[t]
    \centering    
    \begin{minipage}{0.02\columnwidth}
        \centering
        \rotatebox{90}{BER Ratio}
    \end{minipage}%
    \begin{minipage}{1.04\columnwidth}
        \centering
        \begin{subfigure}[b]{\linewidth}
            \centering
            \pgfplotsset{compat=1.17}
\usetikzlibrary{arrows,decorations.markings}
\usetikzlibrary{calc,arrows}

\begin{tikzpicture}[scale=1]    

\definecolor{my_pink}{rgb}{1,0.07,0.65}
\definecolor{my_g}{rgb}{0.09,0.71,0.14}
\definecolor{FEC}{rgb}{0.5, 0.0, 0.13}

\definecolor{fuchsia}{rgb}{0.6,0.4,0.8}

\definecolor{bronze}{rgb}{0.8, 0.5, 0.2}
 \definecolor{green2}{rgb}{0.0, 0.8, 0.6}
 \definecolor{my_b2}{rgb}{0.53,0.64,0.93}

\begin{axis}[ 
    width=1\textwidth,
    height=1.6in,   
    xmin=-4, xmax=21,
    ymin=-0.1, ymax=3,  
    xlabel={}, 
    ylabel shift=1ex,
    ylabel={}, 
    xtick={-4,0,...,20},
    xticklabels={}, 
     ytick={0,1,...,3},
    grid=major,
    grid style={dashed,lightgray!75},
      name=bottomaxis
]

\addplot [color=black,solid,line width=1pt,mark options={solid}]coordinates {
         (30,1) (30,2)};\label{Orig_var1_jit12_black}
\addplot [color=black,dotted,line width=1pt,mark options={solid}]coordinates {
         (30,1) (30,2)};\label{dottedLine}
  

\addplot [color=my_pink,dashed,line width=1.2pt,mark options={solid}] file {./txtData/2PAM_Ratio/AproxOrig_varR01_jit35.txt};

\addplot [color=my_pink,dotted,line width=1.2pt, mark options={solid}] file{./txtData/2PAM_Ratio/UpperOrig_varR01_jit35.txt};


\node [draw,fill=white,anchor= north east,font=\scriptsize] at (rel axis cs:0.99,0.98) {\shortstack[l]{  
\ref{dottedLine} \eqref{eq:BER_Approx2_UpperBound}\\
\ref{Dashed_black} \eqref{eq:BER_Aprox2_ICTON}}};

\node [draw,fill=white,anchor= north east,font=\scriptsize] at (rel axis cs:0.8,0.98) {\shortstack[l]{ 
$\sigma_s=0.35$~m, $\sigma_R^2=0.1$}};


\end{axis} 

\end{tikzpicture}
        \end{subfigure}
        
        \vspace{-0.4em} 
        
        \begin{subfigure}[b]{\linewidth}
            \centering
            \pgfplotsset{compat=1.17}
\usetikzlibrary{arrows,decorations.markings}
\usetikzlibrary{calc,arrows}

\begin{tikzpicture}[scale=1]    

\definecolor{my_pink}{rgb}{1,0.07,0.65}
\definecolor{my_g}{rgb}{0.09,0.71,0.14}
\definecolor{FEC}{rgb}{0.5, 0.0, 0.13}

\definecolor{fuchsia}{rgb}{0.6,0.4,0.8}

\definecolor{bronze}{rgb}{0.8, 0.5, 0.2}
 \definecolor{green2}{rgb}{0.0, 0.8, 0.6}
 \definecolor{my_b2}{rgb}{0.53,0.64,0.93}

\begin{axis}[ 
    width=1\textwidth,
    height=1.6in,   
    xmin=-4, xmax=21,
    ymin=0, ymax=3,  
    xlabel={}, 
    xtick={-4,0,...,20},
    xticklabels={},
    ytick={0,1,2,...,3},
    ylabel shift=1ex,
    ylabel={}, 
    grid=major,
    grid style={dashed,lightgray!75},
      name=bottomaxis
]

\addplot [color=black,dotted,line width=1pt,mark options={solid}]coordinates {
         (30,1) (30,2)};\label{dottedLine}


\addplot [color=orange,dashed,line width=1.2pt,mark options={solid}] file {./txtData/2PAM_Ratio/AproxOrig_varR05_jit25.txt};

\addplot [color=orange,dotted,line width=1.2pt,mark options={solid}] file {./txtData/2PAM_Ratio/UpperOrig_varR05_jit25.txt};


\node [draw,fill=white,anchor= north east,font=\scriptsize] at (rel axis cs:0.99,0.98) {\shortstack[l]{  
\ref{dottedLine} \eqref{eq:BER_Approx2_UpperBound}\\
\ref{Dashed_black} \eqref{eq:BER_Aprox2_ICTON}}};

\node [draw,fill=white,anchor= north east,font=\scriptsize] at (rel axis cs:0.8,0.98) {\shortstack[l]{ 
 $\sigma_s=0.25$~m, $\sigma_R^2=0.5$}};


\end{axis} 

\end{tikzpicture}
        \end{subfigure}
        
        \vspace{-0.4em} 
        
        \begin{subfigure}[b]{\linewidth}
            \centering
            \pgfplotsset{compat=1.17}
\usetikzlibrary{arrows,decorations.markings}
\usetikzlibrary{calc,arrows}

\begin{tikzpicture}[scale=1]    

\definecolor{my_pink}{rgb}{1,0.07,0.65}
\definecolor{my_g}{rgb}{0.09,0.71,0.14}
\definecolor{FEC}{rgb}{0.5, 0.0, 0.13}

\definecolor{fuchsia}{rgb}{0.6,0.4,0.8}

\definecolor{bronze}{rgb}{0.8, 0.5, 0.2}
 \definecolor{green2}{rgb}{0.0, 0.8, 0.6}
 \definecolor{my_b2}{rgb}{0.53,0.64,0.93}

\begin{axis}[ 
    width=1\textwidth,
    height=1.6in,   
    xmin=-4, xmax=21,
    ymin=0, ymax=3,  
    xlabel={Transmit Power, $P$~[dBm]}, 
    ylabel shift=1ex,
    ylabel={ }, 
    ytick={0,1,2,...,3},
    yshift=3ex,
   xtick={-4,0,...,20},
    grid=major,
    grid style={dashed,lightgray!75},
      name=bottomaxis
]


\addplot [color=black,dotted,line width=1pt,mark options={solid}]coordinates {
         (30,1) (30,2)};\label{dottedLine}


\addplot [color=my_g,dashed,line width=1.2pt,mark options={solid}] file {./txtData/2PAM_Ratio/AproxOrig_varR09_jit2.txt};

\addplot [color=my_g,dotted,line width=1.2pt,mark options={solid}] file {./txtData/2PAM_Ratio/UpperOrig_varR09_jit2.txt};


\node [draw,fill=white,anchor= north east,font=\scriptsize] at (rel axis cs:0.99,0.98) {\shortstack[l]{  
\ref{dottedLine} \eqref{eq:BER_Approx2_UpperBound}\\
\ref{Dashed_black} \eqref{eq:BER_Aprox2_ICTON}}};

\node [draw,fill=white,anchor= north east,font=\scriptsize] at (rel axis cs:0.8,0.98) {\shortstack[l]{ 
$\sigma_s=0.2$~m, $\sigma_R^2=0.9$}};


\end{axis} 

\end{tikzpicture}
        \end{subfigure}
    \end{minipage}
    
    \caption{Ratio between BER approximations in \eqref{eq:BER_Approx2_UpperBound} and \eqref{eq:BER_Aprox2_ICTON} and the exact BER in \eqref{eq:ExactOOKBER} versus transmitted power.}
   \label{FigRatio_2PAM}
\end{figure}
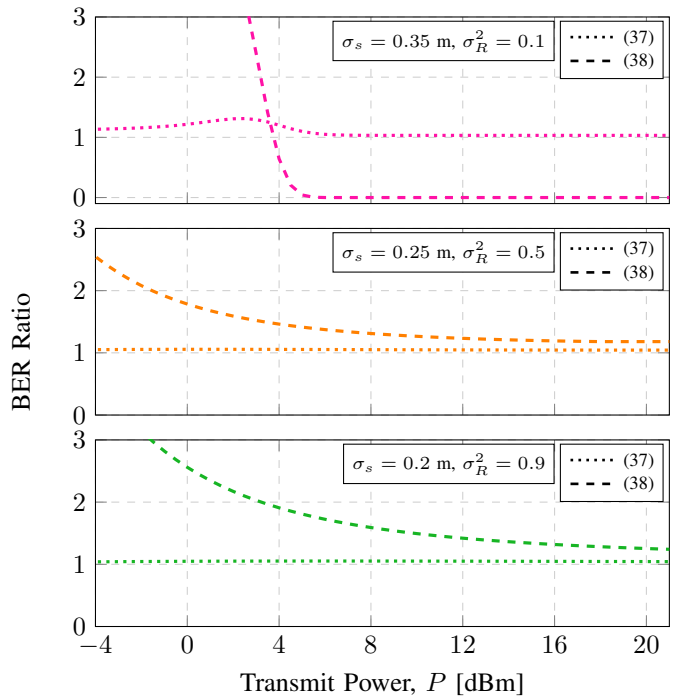

Fig.~\ref{FigRatio_2PAM} provides additional insights into the accuracy of the proposed approximations \eqref{eq:BER_Approx2_UpperBound} and \eqref{eq:BER_Aprox2_ICTON}. This figure depicts the ratio between these BER approximations and the exact BER in \eqref{eq:ExactOOKBER} as a function of the transmitted power. Dotted curves represent the ratio between \eqref{eq:ExactOOKBER} and \eqref{eq:BER_Approx2_UpperBound}, while dashed curves correspond to the ratio between \eqref{eq:ExactOOKBER} and \eqref{eq:BER_Aprox2_ICTON}.
Fig.~\ref{FigRatio_2PAM} shows that operating in the high pointing error regime (pink curves), the ratio between \eqref{eq:ExactOOKBER} and \eqref{eq:BER_Aprox2_ICTON} (dashed curve) converges to $0$. This is explained by the slope of the BER approximation ~\eqref{eq:BER_Aprox2_ICTON}, which differs from the correct slope, as shown in Fig.~\ref{FigBER_OOK}.
In contrast, the ratio between \eqref{eq:ExactOOKBER} and \eqref{eq:BER_Approx2_UpperBound} stabilizes at approximately $1.04$, for all three cases, as the transmitted power increases.\footnote{To obtain the dotted pink curve, simulations were performed up to $10$~dBm and BER$=10^{-9}$. Reaching $20$~dBm would require simulating up to BER$\approx10^{-15}$, so the results have been extrapolated.} 
This indicates that while the approximation \eqref{eq:BER_Approx2_UpperBound} does not exactly match the exact BER, it maintains the same asymptotic decay rate. In other words, the approximation slightly overestimates the BER by a constant factor but captures the correct slope.
The stabilization of the ratio implies that the slopes of the BER and the approximation remain nearly identical in the high-power regime. This suggests that \eqref{eq:BER_Approx2_UpperBound} is asymptotically accurate for predicting the BER decay rate and, in general, provides a more reliable BER estimate.
Table \ref{tab:BER_Comparation} shows a comparison between the exact average BER~\eqref{eq:ExactOOKBER}, and the approximations~\eqref{eq:BER_Approx2_UpperBound} and~\eqref{eq:BER_Aprox2_ICTON} in terms of accuracy, simplicity and slope. 

\begin{table}[t]
\renewcommand{\arraystretch}{1.3}
\centering
\caption{Comparison of average BER Expressions}
\begin{tabular}{ccccc}
\hline
{\textbf{Expression}} & {\textbf{Ref.}} & \textbf{{Accuracy}} & {\textbf{Simplicity}}& {\textbf{Slope}}\\
\hline
\eqref{eq:ExactOOKBER}    & \cite{ICTON} & Exact & Complex & -\\
\eqref{eq:BER_Approx2_UpperBound}& \cite{OECC} & More accurate & Moderately simple & Correct\\
\eqref{eq:BER_Aprox2_ICTON}  & \cite{ICTON}  &  Less accurate & Very simple & Incorrect\\
\hline
\end{tabular}
\label{tab:BER_Comparation}
\end{table}

\begin{figure}[t]
        \centering
        \pgfplotsset{compat=1.17}
\usetikzlibrary{arrows,decorations.markings}
\usetikzlibrary{calc,arrows}

\begin{tikzpicture}[scale=1]    

\definecolor{my_pink}{rgb}{1,0.07,0.65}
\definecolor{my_g}{rgb}{0.09,0.71,0.14}
\definecolor{FEC}{rgb}{0.5, 0.0, 0.13}

\definecolor{fuchsia}{rgb}{0.6,0.4,0.8}

\definecolor{bronze}{rgb}{0.8, 0.5, 0.2}
 \definecolor{green2}{rgb}{0.0, 0.8, 0.6}
 \definecolor{my_b2}{rgb}{0.53,0.64,0.93}

\begin{semilogyaxis}[ 
    width=0.49\textwidth,
    height=3.5in,   
    xmin=0, xmax=25,
    ymin=1e-6, ymax=7e-1,  
    xlabel={Transmit Power, $P$~[dBm]}, 
    ylabel shift=-1ex,
    ylabel={Average SER,~$\overline{P}_{s,2}$}, 
    xtick={0,5,...,25},
    grid=major,
    grid style={dashed,lightgray!75},
      name=bottomaxis
]

\addplot [color=black,solid,line width=1pt,mark options={solid}]coordinates {
         (10,1) (10,2)};\label{Orig_var1_jit12_black}
\addplot [color=black,dashed,line width=1pt,mark=triangle*,mark options={solid,fill=black, scale=0.6}]coordinates {
         (10,3) (10,4)};\label{Dashed_b_triangle}         
\addplot [color=black,dotted,line width=1.2pt,mark options={solid}]coordinates {
         (10,1) (10,2)};\label{dotted_black}
 \addplot [only marks, color=black, mark=*, mark options={solid,scale=0.8, fill=white}]coordinates {(0.35,25)};\label{MC}

\addplot [only marks, color=black, mark=square*, mark options={solid,scale=0.6}]coordinates {(0.35,25)};\label{Square_apr}


\addplot [color=my_pink,solid,line width=1pt,mark options={solid}] file {./txtData/SER_PAM/varR01_jit35/4PAM_SER_orig.txt};

\addplot [color=my_pink,dotted,line width=1.2pt, mark options={solid}] file{./txtData/SER_PAM/UpperBound/varR01_jit35/4PAM_aprox2.txt}; 

\addplot [only marks, color=my_pink, mark=*, mark options={solid,fill=white,scale=0.8}]file{./txtData/SER_PAM/varR01_jit35/4PAM_SER_MC.txt};


 \addplot [color=orange,solid,line width=1pt,mark options={solid}] file {./txtData/SER_PAM/varR05_jit25/4PAM_SER_orig.txt};

\addplot [color=orange,dotted,line width=1.2pt, mark options={solid}]file{./txtData/SER_PAM/UpperBound/varR05_jit25/4PAM_aprox2.txt};

\addplot [only marks, color=orange, mark=*, mark options={solid,fill=white,scale=0.8}]file{./txtData/SER_PAM/varR05_jit25/4PAM_SER_MC.txt};


 \addplot [color=my_g,solid,line width=1pt,mark options={solid}] file {./txtData/SER_PAM/varR09_jit2/4PAM_SER_orig.txt};


\addplot [color=my_g,dotted,line width=1.2pt, mark options={solid}]file{./txtData/SER_PAM/UpperBound/varR09_jit2/4PAM_aprox2.txt};

\addplot [only marks, color=my_g, mark=*, mark options={solid,fill=white,scale=0.8}]file{./txtData/SER_PAM/varR09_jit2/4PAM_SER_MC.txt};


  \addplot [ color=blue, line width=1pt]
   coordinates {(-15,1e-3) (30,1e-3)};
   \label{FEClimitBlue}

\end{semilogyaxis} 

\draw[-,dashed, black,line width=0.5pt] ($(input.north)+(3.57,10.7em)$) -- ($(input.north)+(3.57,9em)$);
\node at ($(input.north) + (3.15, 3.05 cm)$) [scale=1, font=\scriptsize] {$\Delta\approx0.06$ dB};

\draw[-,dashed, black,line width=0.5pt] ($(input.north)+(4.85,10.7em)$) -- ($(input.north)+(4.85,9em)$);
\node at ($(input.north) + (4.9, 3.05cm)$) [scale=1, font=\scriptsize] {$\Delta\approx0.07$ dB};

\draw[-,dashed, black,line width=0.5pt] ($(input.north)+(2.08,10.7em)$) -- ($(input.north)+(2.08,9em)$);
\node at ($(input.north) + (1.5, 3.05cm )$) [scale=1, font=\scriptsize] {$\Delta\approx0.19$ dB};

\node[ellipse, draw,minimum width = 0.6cm, 
	minimum height = 0.05cm, rotate=45] (e) at (2.9,1.3) {};

\node [draw,color=my_pink, fill=white,anchor= south west,font=\scriptsize,text=black] at (1.2,0.3 em) {\shortstack[l]{ 
$\sigma_s=0.35$ m \\
$\sigma_R^2=0.1$}};

\node [draw,color=my_g, fill=white,anchor= south west,font=\scriptsize,text=black] at (5.35,0.3 em) {\shortstack[l]{ 
$\sigma_s=0.2$ m \\
$\sigma_R^2=0.9$}};
\node[ellipse, draw,minimum width = 0.6cm, 
	minimum height = 0.05cm,rotate=45] (e) at (6.4,1.3)  {};

\node[ellipse, draw,minimum width = 0.6cm, 
	minimum height = 0.05cm,rotate=45] (e) at (4.8,1.3) {};
\node [draw,color=orange, fill=white,anchor= south west,font=\scriptsize,text=black] at (3.25,0.3em) {\shortstack[l]{ 
$\sigma_s=0.25$ m \\
$\sigma_R^2=0.5$}};

\node [draw,fill=white,anchor= south west,font=\scriptsize] at (4.77,5.97) {\shortstack[l]{ 
\ref{Orig_var1_jit12_black}  Theorem~\ref{th:SER_exact} \\
\ref{dotted_black} Theorem~\ref{Theorem:SER_Aprox} \\
 \hspace{0.12cm}   \ref{MC} \hspace{0.24cm}  Monte Carlo \\
  \ref{FEClimitBlue} SER threshold}};

\draw[stealth-,solid, black,line width=0.7pt,rounded corners]($(input.north)+(3,17em)$) -- ($(input.north)+(1.45,15.5em)$);
\node at ($(input.north) + (3.65, 17.3 em)$) [scale=1.2, font=\scriptsize] {Increasing};
\node at ($(input.north) + (3.65, 16.6 em)$) [scale=1.2, font=\scriptsize] {Turbulence};


\node [coordinate](input) {};


\end{tikzpicture}
        \vspace{-2ex}
    \caption{Average SER versus transmit power considering 4-PAM, and three different operating points for ($\sigma_s$, $\sigma_R^2$): ($0.2,0.9$), ($0.25,0.5$), and ($0.35,0.1$).}
    \label{FigSER_4PAM}
\end{figure}
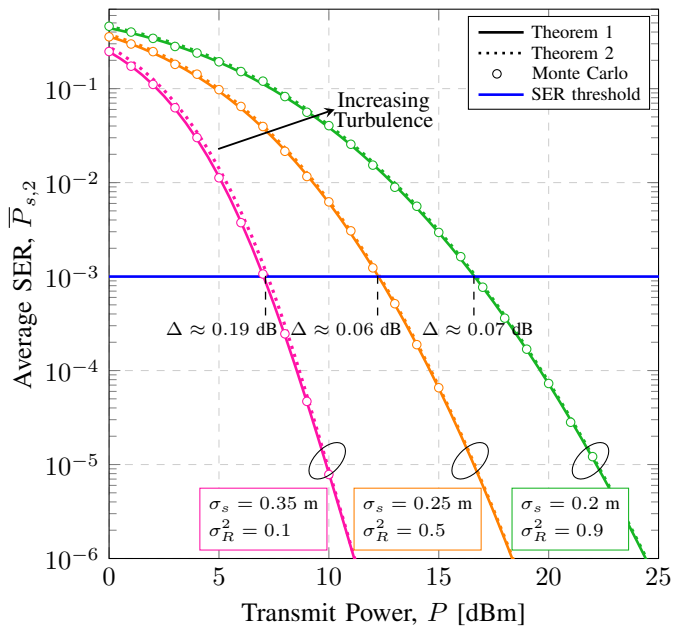


\subsection{Symbol Error Rate}

In Fig.~\ref{FigSER_4PAM}, SERs based on \eqref{eq:SER_exact_PAM} (Theorem~\ref{th:SER_exact}) and \eqref{eq:SER_Approx2_UpperBound} (Theorem~\ref{Theorem:SER_Aprox}) are plotted as a function of the average transmit power $P$ for 4-PAM. 
Results are validated using Monte Carlo simulation (shown with circles), whereas theoretical results are plotted with solid or dotted lines, corresponding to the exact SER in \eqref{eq:SER_exact_PAM} and the approximate SER in~\eqref{eq:SER_Approx2_UpperBound}, resp.
Moreover, an SER threshold of $1\times10^{-3}$ is shown.  
Similarly to Fig.~\ref{FigBER_OOK}, here we use $\Delta$, as the difference in transmit power $P$ at which the exact SER~\eqref{eq:SER_exact_PAM} and our proposed SER approximation \eqref{eq:SER_Approx2_UpperBound} reach the SER threshold.

Similar to Fig.~\ref{FigBER_OOK}, we first observe from Fig.~\ref{FigSER_4PAM} that for all operating points under consideration, a perfect match is shown between the simulated results and those obtained from the derived exact expression.
We then observe that for the first operating point (high pointing
errors and very weak turbulence, pink curves), the SER approximation \eqref{eq:SER_Approx2_UpperBound} is $0.19$~dB away from the exact SER \eqref{eq:SER_exact_PAM} at the SER threshold. 
For the second (medium pointing errors and weak turbulence, orange curves) and third operating points (the lowest pointing errors and most severe weak turbulence, green curves), the gap reduces to
$\Delta\approx 0.06$~dB and $\Delta\approx0.07$~dB, resp. In both cases, the dotted line corresponding to~\eqref{eq:SER_Approx2_UpperBound} is so close to the exact SER that it is barely visible.
From Fig.~\ref{FigSER_4PAM}, we conclude that for weak turbulence and 4-PAM, the proposed average SER approximation \eqref{eq:SER_Approx2_UpperBound} is extremely accurate for the second and third cases, with causing a small ($\leq 0.19$ dB)
inaccuracy for the first case.
The first two rows of Table \ref{tab:SER_Comparation} show a comparison between the exact average SER~\eqref{eq:SER_exact_PAM} and the approximation~\eqref{eq:SER_Approx2_UpperBound} in terms of accuracy and simplicity. 

\begin{table}[t]
\renewcommand{\arraystretch}{1.3}
\centering
\caption{Comparison of average SER Expressions}
    \begin{tabular}{cccc}
        \hline
        {\textbf{Expression}} & {\textbf{Ref.}} & \textbf{{Accuracy}} & {\textbf{Simplicity}}\\
        \hline
        \eqref{eq:SER_exact_PAM}    & This work  & Exact & Complex\\
        \eqref{eq:SER_Approx2_UpperBound}& This work  & Very accurate & Moderately simple \\
        \eqref{eq:SER_DenseMPAM_No4Pi}& This work & Less accurate & Moderately simple \\
        \hline
    \end{tabular}
    \label{tab:SER_Comparation}
\end{table}

In Fig.~\ref{FigSER_VarR01_jit35}, the average SERs in~\eqref{eq:SER_exact_PAM} and \eqref{eq:SER_Approx2_UpperBound} are plotted for $M$-PAM ($M=2,4,8,16,32$) as a function of the transmitted power $P$ in the high pointing errors and very weak turbulence regime ($\sigma_s=0.35$ m and $\sigma_R^2=0.1$). As expected, we observe a degradation in SER as the modulation order $M$ increases. For $P=6$~dBm, increasing the modulation order from $M=2$ (OOK) to $M=4$ increases the SER from $1.4\times10^{-6}$ to $4.9\times10^{-3}$. To guarantee that the SER threshold is achieved for all modulation formats, the transmitted power must be increased as $M$ increases.
However, FSO systems must be designed to be eye-safe, thus they are power-limited.
When the allowable transmit power is $10$~dBm, for a channel affected by pointing errors and very weak turbulence, the system can operate with up to $4$-PAM while still ensuring an SER below the threshold. 
From Fig.~\ref{FigSER_VarR01_jit35}, we can conclude that, as long as eye-safety is ensured, increasing the transmitted power can be used to switch to a higher modulation order for $M$-PAM. 

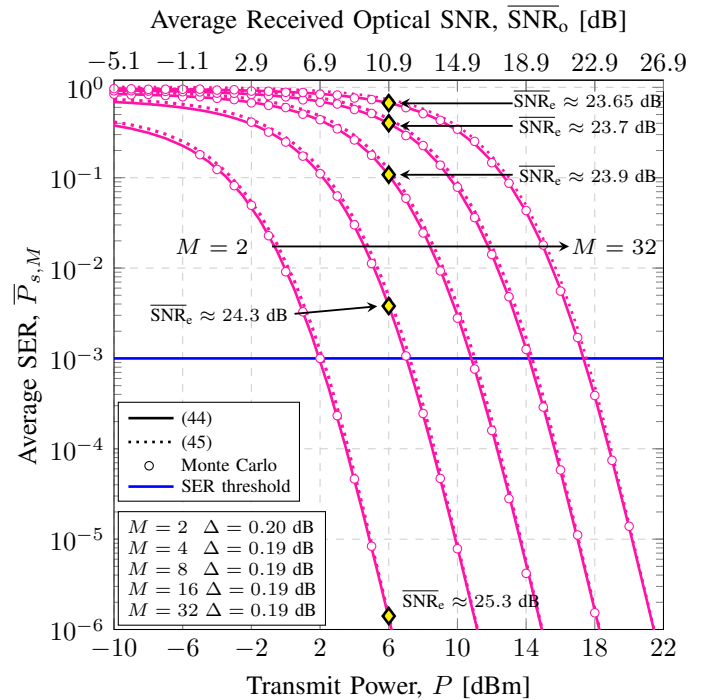
\begin{figure}[t]
        \centering
        \pgfplotsset{compat=1.17}
\usetikzlibrary{arrows,decorations.markings}
\usetikzlibrary{calc,arrows}

\begin{tikzpicture}[scale=1]    

\definecolor{my_pink}{rgb}{1,0.07,0.65}
\definecolor{g2}{rgb}{0,0.5,0.5}
\definecolor{FEC}{rgb}{0.5, 0.0, 0.13}

\definecolor{fuchsia}{rgb}{0.6,0.4,0.8}

\definecolor{bronze}{rgb}{0.8, 0.5, 0.2}
 \definecolor{green2}{rgb}{0.0, 0.8, 0.6}
 \definecolor{my_b2}{rgb}{0.53,0.64,0.93}

\begin{semilogyaxis}[ 
    width=0.49\textwidth,
    height=3.5in,   
    xmin=-10, xmax=22,
    ymin=1e-6, ymax=1.2,  
    xlabel={Transmit Power, $P$~[dBm]}, 
    ylabel shift=-1ex,
    ylabel={Average SER,~$\overline{P}_{s,M}$}, 
    xtick={-10,-6,...,22},
    grid=major,
    grid style = {dashed,lightgray!75},
      name=bottomaxis
]

\addplot [color=black,solid,line width=1pt,mark options={solid}]coordinates {
         (10,10) (10,12)};\label{Orig_var1_jit12_black}      
\addplot [color=black,dotted,line width=1pt,mark options={solid}]coordinates {
         (10,10) (10,12)};\label{dotted_black}
 \addplot [only marks, color=black, mark=*, mark options={solid,scale=0.8, fill=white}]coordinates {(0.35,25)};\label{MC}

\addplot [only marks, color=gray, mark=*, mark options={solid,scale=0.8, fill=gray}]coordinates {(0.35,25)};\label{2}
\addplot [only marks, color=blue, mark=*, mark options={solid,scale=0.8, fill=blue}]coordinates {(0.35,25)};\label{4}
\addplot [only marks, color=red, mark=*, mark options={solid,scale=0.8, fill=red}]coordinates {(0.35,25)};\label{8}
\addplot [only marks, color=black, mark=*, mark options={solid,scale=0.8, fill=black}]coordinates {(0.35,25)};\label{16}
\addplot [only marks, color=g2, mark=*, mark options={solid,scale=0.8, fill=g2}]coordinates {(0.35,25)};\label{32}

\addplot [only marks, color=black, mark=square*, mark options={solid,scale=0.6}]coordinates {(0.35,25)};\label{Square_apr}


\addplot [color=blue, line width=1pt]coordinates {(-15,1e-3) (30,1e-3)};\label{FEClimitBlue}
   

\addplot [color=my_pink,solid,line width=1pt,mark options={solid}] file {./txtData/SER_PAM/varR01_jit35/OOK_SER_orig.txt};


\addplot [color=my_pink,dotted,line width=1.2pt, mark options={solid}] file{./txtData/SER_PAM/UpperBound/varR01_jit35/OOK_aprox2.txt};

\addplot [only marks, color=my_pink, mark=*, mark options={solid,fill=white,scale=0.8}]file{./txtData/SER_PAM/varR01_jit35/OOK_MC.txt};\label{2PAM}

\addplot [ color=my_pink,solid,line width=1pt,mark options={solid}] file {./txtData/SER_PAM/varR01_jit35/4PAM_SER_orig.txt};


\addplot [color=my_pink,dotted,line width=1.2pt, mark options={solid}] file{./txtData/SER_PAM/UpperBound/varR01_jit35/4PAM_aprox2.txt};
\addplot [only marks, color=my_pink, mark=*, mark options={solid,scale=0.8, fill=white}]file{./txtData/SER_PAM/varR01_jit35/4PAM_SER_MC.txt};\label{4PAM}


\addplot [color=my_pink,solid,line width=1pt,mark options={solid}]file{./txtData/SER_PAM/varR01_jit35/8PAM_SER_orig.txt};


\addplot [color=my_pink,dotted,line width=1.2pt, mark options={solid}] file{./txtData/SER_PAM/UpperBound/varR01_jit35/8PAM_aprox2.txt};

\addplot [only marks, color=my_pink, mark=*, mark options={solid,fill=white,scale=0.8}]file{./txtData/SER_PAM/varR01_jit35/8PAM_SER_MC.txt};\label{8PAM}


\addplot [color=my_pink,solid,line width=1pt,mark options={solid}]file{./txtData/SER_PAM/varR01_jit35/16PAM_SER_orig.txt};

\addplot [color=my_pink,dotted,line width=1.2pt, mark options={solid}] file{./txtData/SER_PAM/UpperBound/varR01_jit35/16PAM_aprox2.txt};

\addplot [only marks, color=my_pink, mark=*, mark options={solid,fill=white,scale=0.8}]file{./txtData/SER_PAM/varR01_jit35/16PAM_SER_MC.txt};\label{16PAM}


\addplot [color=my_pink,solid,line width=1pt,mark options={solid}]file{./txtData/SER_PAM/varR01_jit35/32PAM_SER_orig.txt};

\addplot [color=my_pink,dotted,line width=1pt, mark options={solid}] file{./txtData/SER_PAM/UpperBound/varR01_jit35/32PAM_aprox2.txt};

\addplot [only marks, color=my_pink, mark=*, mark options={solid,fill=white,scale=0.8}]file{./txtData/SER_PAM/varR01_jit35/32PAM_SER_MC.txt};\label{32PAM}

\addplot [color=black, mark=diamond*, mark options={fill=yellow, line width=1pt}, mark size=3pt] coordinates {(6,1.409e-06)};

\addplot [color=black, mark=diamond*, mark options={fill=yellow, line width=1pt}, mark size=3pt]  coordinates {(6,3.8e-3)};

\addplot [color=black, mark=diamond*, mark options={fill=yellow, line width=1pt}, mark size=3pt] coordinates {(6,0.108)};

\addplot [color=black, mark=diamond*, mark options={fill=yellow, line width=1pt}, mark size=3pt]  coordinates {(6,4.01e-1)};

\addplot [color=black, mark=diamond*, mark options={fill=yellow, line width=1pt}, mark size=3pt]  coordinates {(6,6.7e-1)};

\end{semilogyaxis} 
         
\begin{axis}[
        width=0.49\textwidth,
        height=3.5in,   
        xlabel={{Average} Received Optical SNR, $\overline{\text{SNR}}_\text{o}$ [dB]},
        xmin=-5.1,xmax=26.9,
        axis y line=none,
        axis x line*=top,
        xtick={-5.1,-1.1,2.9,...,26.9},
        ] 
         \addplot[transparent] coordinates {(3.2,0)(12.8,1)};
\end{axis}

\node [draw,fill=white,anchor= south west,font=\scriptsize] at (0.05,1.7) {\shortstack[l]{ 
\ref{Orig_var1_jit12_black} \eqref{eq:SER_exact_PAM} \\
\ref{dotted_black}~\eqref{eq:SER_Approx2_UpperBound} \\ 
 \hspace{0.15cm}   \ref{MC} \hspace{0.2cm}  Monte Carlo\\
  \ref{FEClimitBlue} SER threshold}};

  \node [draw,fill=white,anchor= south west,font=\scriptsize] at (0.05,0.05) {\shortstack[l]{ 
  $M=2$ \hspace{0.025cm} $\Delta=0.20$~dB\\
  $M=4$  \hspace{0.025cm} $\Delta=0.19$~dB\\
 $M=8$  \hspace{0.025cm} $\Delta=0.19$~dB\\
  $M=16$ $\Delta=0.19$~dB\\
  $M=32$ $\Delta=0.19$~dB}};

\node [coordinate](input) {};



\node at ($(input.north) + (4.75cm, 0.4cm )$) [scale=1, font=\scriptsize] {$\overline{\text{SNR}}_\text{e}\approx25.3$~dB};

\node at ($(input.north) + (1.4cm, 4.2cm )$) [scale=1, font=\scriptsize] {$\overline{\text{SNR}}_\text{e}\approx24.3$~dB};

\node at ($(input.north) + (6.3cm, 6.05cm )$) [scale=1, font=\scriptsize] {$\overline{\text{SNR}}_\text{e}\approx23.9$~dB};

\node at ($(input.north) + (6.3cm, 6.7cm )$) [scale=1, font=\scriptsize] {$\overline{\text{SNR}}_\text{e}\approx23.7$~dB};

\node at ($(input.north) + (6.3cm, 7.05cm )$) [scale=1, font=\scriptsize] {$\overline{\text{SNR}}_\text{e}\approx23.65$~dB};

\node at ($(input.north) + (1.3, 14.5em)$) [scale=1.2, font=\scriptsize] {$M=2$};
\node at ($(input.north) + (6.65, 14.5em)$) [scale=1.2, font=\scriptsize] {$M=32$};

\draw[stealth-,solid, black,line width=0.7pt,rounded corners] ($(input.north)+(6.05,14.5em)$) -- ($(input.north)+(2.1,14.5em)$);

\draw[-stealth,solid, black,line width=0.7pt,rounded corners] ($(input.north)+(5.3,6.05)$) -- ($(input.north)+(3.8,6.05)$);

\draw[-stealth,solid, black,line width=0.7pt,rounded corners] ($(input.north)+(5.3,6.7)$) -- ($(input.north)+(3.8,6.7)$);
\draw[-stealth,solid, black,line width=0.7pt,rounded corners] ($(input.north)+(5.3,7)$) -- ($(input.north)+(3.8,7)$);

\draw[-stealth,solid, black,line width=0.7pt,rounded corners] ($(input.north)+(2.4,4.2)$) -- ($(input.north)+(3.5,4.3)$);

\end{tikzpicture}
        \vspace{-2ex}
    \caption{Average SER versus transmit power for $\sigma_s=0.35$~m and $\sigma_R^2=0.1$ and different values of $M$ ($M=2,4,8,16,32$).}
    \label{FigSER_VarR01_jit35}
\end{figure}

An additional x-axis with the average received optical SNR ($\overline{\text{SNR}}_\text{o}$) in \eqref{SNR_op_av} has been included on top of Fig. \ref{FigSER_VarR01_jit35}. At $P=6$~dBm, five different points are shown, each indicating the corresponding average received electrical SNR ($\overline{\text{SNR}}_\text{e}$) as defined in \eqref{SNR_el_av}. As we mentioned in Sec.~\ref{Sec:SystemModel}, the signal power $ \mathbb{E}[X^2]$ is modulation format-dependent, and consequently, the electrical SNR also varies with the modulation order (see \eqref{SNR_el_av}). For a given $P$, there is a unique value of $\overline{\text{SNR}}_\text{o}$, while 
$\overline{\text{SNR}}_\text{e}$ varies depending on $M$. Therefore, a common $\overline{\text{SNR}}_\text{e}$ axis cannot be considered. 
Although the electrical SNR is commonly used, its dependence on the modulation format led some studies\textemdash including this one\textemdash to define and use the optical SNR instead. 



\begin{figure}[t]
        \centering
        \pgfplotsset{compat=1.17}
\usetikzlibrary{arrows,decorations.markings}
\usetikzlibrary{calc,arrows}

\begin{tikzpicture} [scale=1]

\definecolor{my_blue}{rgb}{0.87,0.92,0.98}
\definecolor{my_b}{rgb}{0.67,0.75,0.98}
\definecolor{my_b2}{rgb}{0.53,0.64,0.93}
\definecolor{my_g}{rgb}{0.09,0.71,0.14}
\definecolor{fuchsia}{rgb}{0.6,0.4,0.8}
\definecolor{my_pink}{rgb}{1,0.07,0.65}
\definecolor{FEC}{rgb}{0.5, 0.0, 0.13}

\begin{axis}[
    width=0.49\textwidth,
    height=3.1in,   
xmin = 2,
xmax = 9,
ymin = 2.75,
ymax = 5.5,
xlabel = {$m$},
ylabel = {Power Increase, $\Delta P$ [dB]},
ylabel shift=1ex,
grid = major,
grid style = {dashed,lightgray!75},
ytick = {3,3.5,...,5.5},
xtick={2,3,...,9}
]
 \addplot [only marks, color=my_g, mark=*, mark options={solid,scale=0.8}]coordinates {(0.35,1)};\label{Cgreen}
  \addplot [only marks, color=orange, mark=*, mark options={solid,scale=0.8}]coordinates {(0.35,1)};\label{Corange} 
  \addplot [only marks, color=my_pink, mark=*, mark options={solid,scale=0.8}]coordinates {(0.35,1)};\label{Cpink}


\addplot [color=my_pink,solid,line width=1pt,mark options={solid}]
   coordinates {(2,5.0594) (3,3.791) (4,3.3528) (5,3.1777) (6,3.089) (7,3.0485) (8,3.0295) (9,3.0195)}; 

\addplot  [only marks, color=my_pink, mark=*, mark options={solid,fill=my_pink,scale=0.8}]coordinates {(2,5.0594) (3,3.791) (4,3.3528) (5,3.1777) (6,3.089) (7,3.0485) (8,3.0295) (9,3.0195)};

\addplot [color=orange,solid,line width=1pt,mark options={solid}]
   coordinates {
         (2,5.2330) (3,3.8515) (4,3.387) (5,3.1875) (6,3.097) (7,3.054) (8,3.032) (9,3.022)};
         
\addplot [only marks, color=orange, mark=*, mark options={solid,fill=orange,scale=0.8}]coordinates {
         (2,5.2330) (3,3.8515) (4,3.387) (5,3.1875) (6,3.097) (7,3.054) (8,3.032) (9,3.022)};

\addplot [color=my_g,solid,line width=1pt,mark options={solid}]
   coordinates {
         (2,5.359) (3,3.898) (4,3.408) (5,3.196) (6,3.102) (7,3.056) (8,3.033) (9,3.023)};
         
\addplot [only marks, color=my_g, mark=*, mark options={solid,fill=my_g,scale=0.8}]coordinates {
      (2,5.359) (3,3.898) (4,3.408) (5,3.196) (6,3.102) (7,3.056) (8,3.033) (9,3.023)};

  \addplot [ color=blue, line width=1pt]
   coordinates {(1,3) (10,3)};
   \label{limit3dB}

\end{axis}

\node at (1.2,0.8) [scale=1.2, font=\scriptsize] {Corollary~\ref{Corollary1}};

\node [draw,fill=white,anchor= south west,font=\scriptsize] at (5.4,5.1){\shortstack[l]{ @SER$=10^{-5}$}};

\node [draw,fill=white,anchor= south west,font=\scriptsize] at (4,4.9) {\shortstack[l]{ 
\ref{Cgreen} \hspace{0.01cm} $\sigma_s=0.2$~m, $\sigma_R^2=0.9$\\
\ref{Corange} \hspace{0.01cm} $\sigma_s=0.25$~m, $\sigma_R^2=0.5$\\
\ref{Cpink} \hspace{0.01cm}  $\sigma_s=0.35$~m, $\sigma_R^2=0.1$}};

\end{tikzpicture}%
        \vspace{-2ex}
    \caption{Power increment required to go from $M$-PAM to $2M$-PAM based on target SER of $10^{-3}$ for three operating points of \mbox{($\sigma_s$, $\sigma_R^2$).}}
    \label{Fig:SER_GAP}
\end{figure}
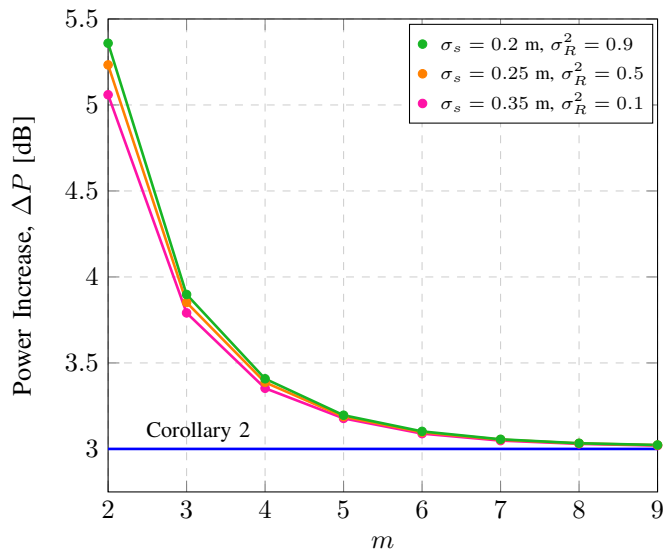

In Fig. \ref{Fig:SER_GAP}, we present the additional transmit power required when the spectral efficiency increases by 1 bit/symbol (i.e, when the modulation order $M$ is doubled) as a function of the number of bits per symbol $m$ for the SER target of $10^{-3}$. 
The power increase when increasing the spectral efficiency by $1$ bit/symbol stabilizes at $3$~dB, as we proved in Corollary~\ref{Corollary1}. From this result it can be concluded that $3$~dB of power increase apply to fading channels.

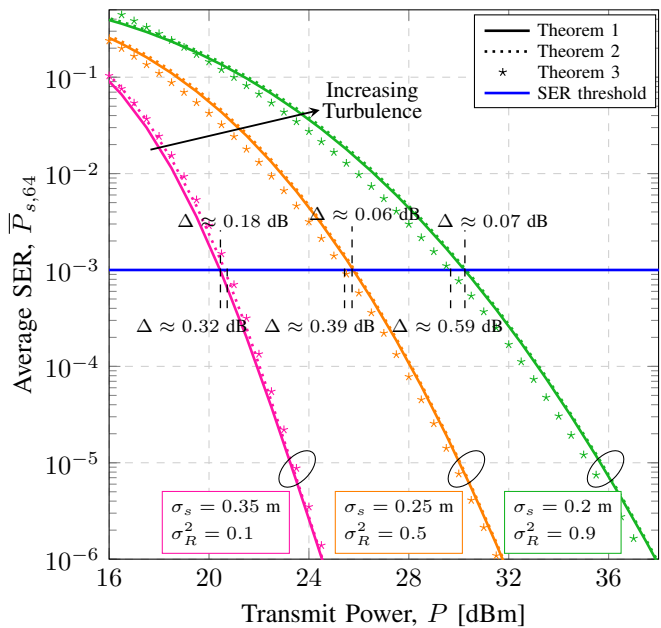
\begin{figure}[t]
        \centering
        \pgfplotsset{compat=1.17}
\usetikzlibrary{arrows,decorations.markings}
\usetikzlibrary{calc,arrows}

\begin{tikzpicture}[scale=1]    

\definecolor{my_pink}{rgb}{1,0.07,0.65}
\definecolor{my_g}{rgb}{0.09,0.71,0.14}
\definecolor{FEC}{rgb}{0.5, 0.0, 0.13}

\definecolor{fuchsia}{rgb}{0.6,0.4,0.8}

\definecolor{bronze}{rgb}{0.8, 0.5, 0.2}
 \definecolor{green2}{rgb}{0.0, 0.8, 0.6}
 \definecolor{my_b2}{rgb}{0.53,0.64,0.93}

\begin{semilogyaxis}[ 
    width=0.49\textwidth,
    height=3.5in,   
    xmin=16, xmax=38,
    ymin=1e-6, ymax=5e-1,  
    xlabel={Transmit Power, $P$~[dBm]}, 
    ylabel shift=-1ex,
    ylabel={Average SER,~$\overline{P}_{s,64}$},
    xtick={16,20,...,38},
    grid=major,
    grid style={dashed,lightgray!75},
      name=bottomaxis
]

\addplot [color=black,solid,line width=1pt,mark options={solid}]coordinates {
         (10,1) (10,2)};\label{Orig_var1_jit12_black}
\addplot [color=black,dotted,line width=1pt,mark options={solid}]coordinates {
         (10,1) (10,2)};\label{dotted_black}
 \addplot [only marks, color=black, mark=star, mark options={solid,scale=0.8, fill=white}]coordinates {(0.35,25)};\label{NoPi}
  \addplot [only marks, color=black, mark=triangle*, mark options={solid,scale=0.8, fill=black}]coordinates {(0.35,25)};\label{NoTerm}


\addplot [color=my_pink,solid,line width=1pt,mark options={solid}] file {./txtData/64PAM_SER/True_varR01_jit35.txt};

\addplot [color=my_pink,dotted,line width=1pt, mark options={solid}] file{./txtData/64PAM_SER/Aprox_varR01_jit35.txt};

\addplot [only marks, color= my_pink, mark=star, mark options={solid,fill=white,scale=0.8}]file{./txtData/64PAM_SER/NoPi_varR01_jit35.txt};


\addplot [color=orange,solid,line width=1pt,mark options={solid}] file {./txtData/64PAM_SER/True_varR05_jit25.txt};

\addplot [color=orange,dotted,line width=1pt, mark options={solid}] file{./txtData/64PAM_SER/Aprox_varR05_jit25.txt};

\addplot [only marks, color= orange, mark=star, mark options={solid,fill=white,scale=0.8}]file{./txtData/64PAM_SER/NoPi_varR05_jit25.txt};

\addplot [color=my_g,solid,line width=1pt,mark options={solid}] file {./txtData/64PAM_SER/True_varR09_jit2.txt};

\addplot [color=my_g,dotted,line width=1pt, mark options={solid}] file{./txtData/64PAM_SER/Aprox_varR09_jit2.txt};

\addplot [only marks, color= my_g, mark=star, mark options={solid,fill=white,scale=0.8}]file{./txtData/64PAM_SER/NoPi_varR09_jit2.txt};


  \addplot [ color=blue, line width=1pt]
   coordinates {(10,1e-3) (40,1e-3)};

\end{semilogyaxis} 


\draw[-,dashed, black,line width=0.5pt] ($(input.north)+(3.13,11em)$) -- ($(input.north)+(3.13,9.5em)$);
\draw[-,dashed, black,line width=0.5pt] ($(input.north)+(3.23,11em)$) -- ($(input.north)+(3.23,9.5em)$);
\node at ($(input.north) + (2.8, 3.1 cm)$) [scale=1, font=\scriptsize] {$\Delta\approx0.39$ dB};

\draw[-,dashed, black,line width=0.5pt] ($(input.north)+(4.54,11em)$) -- ($(input.north)+(4.54,9.5em)$);
\draw[-,dashed, black,line width=0.5pt] ($(input.north)+(4.73,11em)$) -- ($(input.north)+(4.73,9.5em)$);

\node at ($(input.north) + (4.5, 3.1cm)$) [scale=1, font=\scriptsize] {$\Delta\approx0.59$ dB};

\draw[-,dashed, black,line width=0.5pt] ($(input.north)+(1.48,11em)$) -- ($(input.north)+(1.48,9.5em)$);
\draw[-,dashed, black,line width=0.5pt] ($(input.north)+(1.57,11em)$) -- ($(input.north)+(1.57,9.5em)$);
\node at ($(input.north) + (1.1, 3.1cm )$) [scale=1, font=\scriptsize] {$\Delta\approx0.32$ dB};


\draw[-,dashed, black,line width=0.5pt] ($(input.north)+(3.23,11.2em)$) -- ($(input.north)+(3.23,12.6em)$);
\node at ($(input.north) + (3.4, 4.6 cm)$) [scale=1, font=\scriptsize] {$\Delta\approx0.06$ dB};

\draw[-,dashed, black,line width=0.5pt] ($(input.north)+(4.73,11.2em)$) -- ($(input.north)+(4.73,12.45em)$);
\node at ($(input.north) + (5.1, 4.5cm)$) [scale=1, font=\scriptsize] {$\Delta\approx0.07$ dB};

\draw[-,dashed, black,line width=0.5pt] ($(input.north)+(1.48,11.2em)$) -- ($(input.north)+(1.48,12.45em)$);
\node at ($(input.north) + (1.65, 4.5cm )$) [scale=1, font=\scriptsize] {$\Delta\approx0.18$ dB};


\node[ellipse, draw,minimum width = 0.6cm, 
	minimum height = 0.05cm, rotate=45] (e) at (2.5,1.2) {};

\node [draw,color=my_pink, fill=white,anchor= south west,font=\scriptsize,text=black] at (0.7,0.2 em) {\shortstack[l]{ 
$\sigma_s=0.35$ m \\
$\sigma_R^2=0.1$}};
 
\node[ellipse, draw,minimum width = 0.6cm, 
	minimum height = 0.05cm,rotate=45] (e) at (6.6,1.2)  {};
\node [draw,color=my_g, fill=white,anchor= south west,font=\scriptsize,text=black] at (5.25,0.2 em) {\shortstack[l]{ 
$\sigma_s=0.2$ m \\
$\sigma_R^2=0.9$}};

\node[ellipse, draw,minimum width = 0.6cm, 
	minimum height = 0.05cm,rotate=45] (e) at (4.75,1.2) {};
    
\node [draw,color=orange, fill=white,anchor= south west,font=\scriptsize,text=black] at (3,0.2 em) {\shortstack[l]{ 
$\sigma_s=0.25$ m \\
$\sigma_R^2=0.5$}};

\node [draw,fill=white,anchor= south west,font=\scriptsize] at (4.85,6) {\shortstack[l]{ 
\ref{Orig_var1_jit12_black} Theorem \ref{th:SER_exact} \\
\ref{dotted_black} Theorem \ref{Theorem:SER_Aprox} \\
 \hspace{0.12cm}  \ref{NoPi} \hspace{0.24cm} Theorem \ref{Theorem:SER_NoPi} \\
  \ref{FEClimitBlue} SER threshold}};

\draw[stealth-,solid, black,line width=0.7pt,rounded corners]($(input.north)+(2.8,17em)$) -- ($(input.north)+(0.55,15.5em)$);
\node at ($(input.north) + (3.5, 17.7 em)$) [scale=1.2, font=\scriptsize] {Increasing};
\node at ($(input.north) + (3.5, 17 em)$) [scale=1.2, font=\scriptsize] {Turbulence};


\node [coordinate](input) {};


\end{tikzpicture}
        \vspace{-2ex}
    \caption{Average SER versus transmit power considering $64$-PAM, and three different operating points for ($\sigma_s$, $\sigma_R^2$): ($0.2,0.9$), ($0.25,0.5$), and ($0.35,0.1$).}
    \label{FigSER_64PAM}
\end{figure}
 
In Fig.~\ref{FigSER_64PAM}, SERs in \eqref{eq:SER_exact_PAM} (Theorem~\ref{th:SER_exact}), \eqref{eq:SER_Approx2_UpperBound} (Theorem~\ref{Theorem:SER_Aprox}) and \eqref{eq:SER_DenseMPAM_No4Pi} (Theorem~\ref{Theorem:SER_NoPi}) are plotted as a function of the average transmit power $P$ for $64$-PAM. 
Solid and dotted curves correspond to the exact SER~\eqref{eq:SER_exact_PAM} and the approximate SER in~\eqref{eq:SER_Approx2_UpperBound}, resp., while the simpler SER approximation in \eqref{eq:SER_DenseMPAM_No4Pi} is shown with stars. 
For the SER values of interest (i.e., SERs below $10^{-2}$), we can observe that the approximation proposed in \eqref{eq:SER_DenseMPAM_No4Pi} is less accurate compared to \eqref{eq:SER_Approx2_UpperBound}: it has a gap of $\Delta=0.59$~dB to the actual SER, for the third case, while for the second and first cases, the gap $\Delta$ is below $0.39$~dB. However, this expression provides a good trade-off between complexity and accuracy. 

\begin{figure}[t!]
    \centering    
    \begin{minipage}{0.02\columnwidth}
        \centering
        \rotatebox{90}{SER Ratio}
    \end{minipage}%
    \begin{minipage}{1.04\columnwidth}
        \centering
        \begin{subfigure}[b]{1\columnwidth}
            \centering
            \pgfplotsset{compat=1.17}
\usetikzlibrary{arrows,decorations.markings}
\usetikzlibrary{calc,arrows}

\begin{tikzpicture}[scale=1]    

\definecolor{my_pink}{rgb}{1,0.07,0.65}
\definecolor{my_g}{rgb}{0.09,0.71,0.14}
\definecolor{FEC}{rgb}{0.5, 0.0, 0.13}

\definecolor{fuchsia}{rgb}{0.6,0.4,0.8}

\definecolor{bronze}{rgb}{0.8, 0.5, 0.2}
 \definecolor{green2}{rgb}{0.0, 0.8, 0.6}
 \definecolor{my_b2}{rgb}{0.53,0.64,0.93}

\begin{axis}[ 
    width=1\textwidth,
    height=1.6in,   
     xmin=16, xmax=38,
    ymin=0, ymax=2,  
    xlabel={}, 
    ylabel shift=1ex,
     ylabel={},
    xticklabels={}, 
    grid=major,
    grid style={dashed,lightgray!75},
      name=bottomaxis
]

\addplot [only marks, color= black, mark=pentagon*, mark options={solid,fill=white,scale=1}]coordinates {(10,1) (10,2)};\label{Pentagon}

\addplot [only marks, color= black, mark=star, mark options={solid,fill=black,scale=1}]coordinates {
         (10,1) (10,2)};\label{NoPi_das}
\addplot [only marks, color= black, mark=triangle*, mark options={solid,fill=white,scale=1}]coordinates {
         (10,1) (10,2)};\label{NoTerm_dot}


\addplot [color=my_pink,dotted,line width=1pt,mark options={solid}] 
file {./txtData/64PAM_Ratio/AproxOrig_varR01_jit35.txt};

\addplot [only marks, color= my_pink, mark=star, mark options={solid,fill=my_pink,scale=1}]file{./txtData/64PAM_Ratio/Nopi_Orig_varR01_jit35.txt};



\node [draw,fill=white,anchor= north east,font=\scriptsize] at (rel axis cs:0.99,0.34) {\shortstack[l]{ 
\ref{dotted_black} \eqref{eq:SER_Approx2_UpperBound} \\
\hspace{0.25cm}\ref{NoPi_das} \hspace{0.16cm} \eqref{eq:SER_DenseMPAM_No4Pi}}};

\node [draw,fill=white,anchor= north east,font=\scriptsize] at (rel axis cs:0.8,0.25)  {\shortstack[l]{ 
$\sigma_s=0.35$~m, $\sigma_R^2=0.1$}};

\end{axis} 

\end{tikzpicture}
        \end{subfigure}
        
        \vspace{-0.3em} 
        
        \begin{subfigure}[b]{1\columnwidth}
            \centering
            \pgfplotsset{compat=1.17}
\usetikzlibrary{arrows,decorations.markings}
\usetikzlibrary{calc,arrows}

\begin{tikzpicture}[scale=1]    

\definecolor{my_pink}{rgb}{1,0.07,0.65}
\definecolor{my_g}{rgb}{0.09,0.71,0.14}
\definecolor{FEC}{rgb}{0.5, 0.0, 0.13}

\definecolor{fuchsia}{rgb}{0.6,0.4,0.8}

\definecolor{bronze}{rgb}{0.8, 0.5, 0.2}
 \definecolor{green2}{rgb}{0.0, 0.8, 0.6}
 \definecolor{my_b2}{rgb}{0.53,0.64,0.93}

\begin{axis}[ 
    width=1\textwidth,
    height=1.6in,   
     xmin=16, xmax=38,
    ymin=0, ymax=2,  
    xlabel={}, 
    xticklabels={}, 
    ytick={0,0.5,...,2},
    ylabel shift=1ex,
   ylabel={},
    grid=major,
    grid style={dashed,lightgray!75},
      name=bottomaxis
]

\addplot [only marks, color= black, mark=pentagon*, mark options={solid,fill=white,scale=1}]coordinates {(10,1) (10,2)};\label{Pentagon}

\addplot [only marks, color= black, mark=star, mark options={solid,fill=black,scale=1}]coordinates {
         (10,1) (10,2)};\label{NoPi_das}
\addplot [only marks, color= black, mark=triangle*, mark options={solid,fill=white,scale=1}]coordinates {
         (10,1) (10,2)};\label{NoTerm_dot}


\addplot [color=orange,dotted,line width=1pt,mark options={solid}] 
file {./txtData/64PAM_Ratio/AproxOrig_varR05_jit25.txt};

\addplot [only marks, color= orange, mark=star, mark options={solid,fill=my_pink,scale=1}]file{./txtData/64PAM_Ratio/Nopi_Orig_varR05_jit25.txt};




\node [draw,fill=white,anchor= north east,font=\scriptsize] at (rel axis cs:0.99,0.98) {\shortstack[l]{ 
\ref{dotted_black} \eqref{eq:SER_Approx2_UpperBound} \\
\hspace{0.25cm}\ref{NoPi_das} \hspace{0.16cm} \eqref{eq:SER_DenseMPAM_No4Pi}}};

\node [draw,fill=white,anchor= north east,font=\scriptsize] at (rel axis cs:0.8,0.98)  
{\shortstack[l]{ 
$\sigma_s=0.25$~m, $\sigma_R^2=0.5$}};

\end{axis} 

\end{tikzpicture}
        \end{subfigure}
        
        \vspace{-0.3em} 
        
        \begin{subfigure}[b]{1\columnwidth}
            \centering
            \pgfplotsset{compat=1.17}
\usetikzlibrary{arrows,decorations.markings}
\usetikzlibrary{calc,arrows}

\begin{tikzpicture}[scale=1]    

\definecolor{my_pink}{rgb}{1,0.07,0.65}
\definecolor{my_g}{rgb}{0.09,0.71,0.14}
\definecolor{FEC}{rgb}{0.5, 0.0, 0.13}

\definecolor{fuchsia}{rgb}{0.6,0.4,0.8}

\definecolor{bronze}{rgb}{0.8, 0.5, 0.2}
 \definecolor{green2}{rgb}{0.0, 0.8, 0.6}
 \definecolor{my_b2}{rgb}{0.53,0.64,0.93}

\begin{axis}[ 
    width=1\textwidth,
    height=1.6in,   
    xmin=16, xmax=38,
    ymin=0, ymax=2,  
    xlabel={Transmit Power, $P$~[dBm]}, 
    ylabel shift=1ex,
    ylabel={ },
    xtick={16,20,...,38},
    grid=major,
    grid style={dashed,lightgray!75},
      name=bottomaxis
]

\addplot [only marks, color= black, mark=pentagon*, mark options={solid,fill=white,scale=1}]coordinates {(10,1) (10,2)};\label{Pentagon}

\addplot [only marks, color= black, mark=star, mark options={solid,fill=black,scale=1}]coordinates {
         (10,1) (10,2)};\label{NoPi_das}
\addplot [only marks, color= black, mark=triangle*, mark options={solid,fill=white,scale=1}]coordinates {
         (10,1) (10,2)};\label{NoTerm_dot}



\addplot [color=my_g,dotted,line width=1pt,mark options={solid}] 
file {./txtData/64PAM_Ratio/AproxOrig_varR09_jit2.txt};

\addplot [only marks, color= my_g, mark=star, mark options={solid,fill=white,scale=1}] file {./txtData/64PAM_Ratio/Nopi_Orig_varR09_jit2.txt};



\node [draw,fill=white,anchor= north east,font=\scriptsize] at (rel axis cs:0.99,0.98) {\shortstack[l]{ 
\ref{dotted_black} \eqref{eq:SER_Approx2_UpperBound} \\
\hspace{0.25cm}\ref{NoPi_das} \hspace{0.16cm} \eqref{eq:SER_DenseMPAM_No4Pi}}};

\node [draw,fill=white,anchor= north east,font=\scriptsize] at (rel axis cs:0.8,0.98)  
{\shortstack[l]{ 
$\sigma_s=0.2$~m, $\sigma_R^2=0.9$}};

\end{axis} 

\node [coordinate](input) {};


\end{tikzpicture}
        \end{subfigure}
    \end{minipage}
    
    \caption{Ratio between SER approximations in \eqref{eq:SER_Approx2_UpperBound} and \eqref{eq:SER_DenseMPAM_No4Pi} and the exact SER in \eqref{eq:SER_exact_PAM} versus transmitted power considering $64$-PAM, and three different operating points for \mbox{($\sigma_s$, $\sigma_R^2$):} ($0.2,0.9$), ($0.25,0.5$), and ($0.35,0.1$).}
    \label{FigRatio_64PAM}
\end{figure}
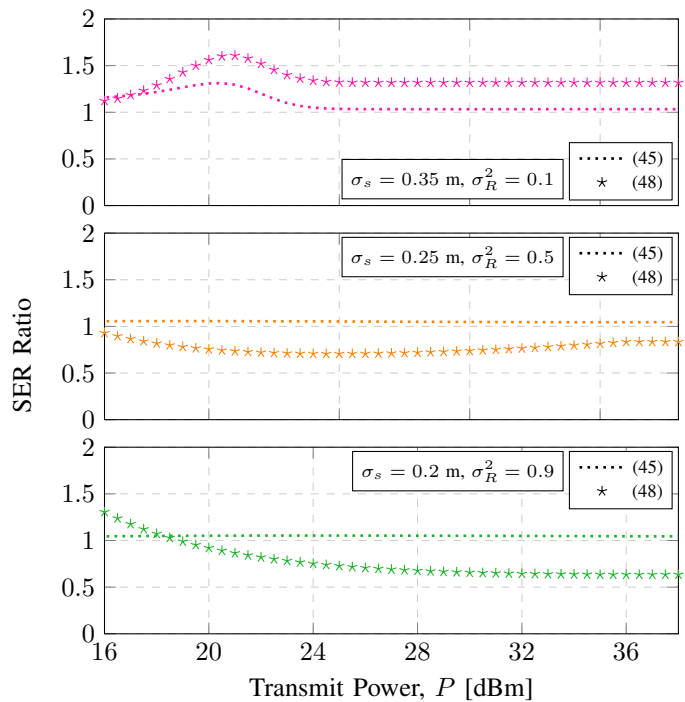

Finally, to examine the accuracy of the different SER approximations in detail, we present in Fig.~\ref{FigRatio_64PAM} the ratio between the two proposed approximations \eqref{eq:SER_Approx2_UpperBound} and \eqref{eq:SER_DenseMPAM_No4Pi}, and the exact expression \eqref{eq:SER_exact_PAM}.
From this figure, we see that in the high-power regime, \eqref{eq:SER_Approx2_UpperBound} is a more accurate approximation than \eqref{eq:SER_DenseMPAM_No4Pi}, as the ratio between \eqref{eq:SER_Approx2_UpperBound} and \eqref{eq:SER_exact_PAM} stabilizes at approximately $1.03$. The ratio between \eqref{eq:SER_DenseMPAM_No4Pi} and \eqref{eq:SER_exact_PAM} is close to the previous one, stabilizing at approximately $1.32$, $0.83$, and $0.63$ for the first, the second and the third cases considered, resp.\footnote{Note that to obtain the pink curves we ran the simulation down to $28.5$~dBm and BER $\approx 10^{-9}$. Reaching $36$~dBm would require simulating up to BER $\approx10^{-15}$, so the results have been extrapolated.}
These results indicate that while the approximation in Theorem~\ref{Theorem:SER_NoPi} does not exactly match the exact SER, it maintains the same asymptotic decay rate. In other words, the approximation is slightly offset from the SER by a constant factor but captures the correct slope.

  

\section{Conclusions}
\label{sec:Conclusions}
In this paper, we developed simple BER and SER approximations for $M$-PAM over an FSO channel that jointly models weak atmospheric turbulence, pointing errors, geometric spread, and atmospheric losses. 
Our proposed expressions provide computationally more efficient solutions to those available in the literature. 
The numerical results presented in this paper show that for the different regimes of turbulence and pointing errors considered in this work, the proposed BER and SER approximations are in close agreement with the exact expressions. 
The proposed approximations were shown to be useful when analyzing the required power increase when dense PAM constellations are used. 
Future work includes an experimental validation of the results presented in this paper as well as extensions to other atmospheric turbulence models such as the Gamma-Gamma distribution.

\bibliographystyle{IEEEtran}
\bibliography{references}


\vspace{11pt}

\begin{IEEEbiographynophoto}{Carmen Álvarez Roa}
was born in Úbeda, Spain, in 1999. She received her B.Sc. and M.Sc. degrees from University of Málaga, both in Telecommunications Engineering, in 2021 and 2023, respectively. Since March 2024, she is a PhD student at the ICT Lab at the Signal Processing Systems (SPS) Group, Eindhoven University of Technology (TU/e), the Netherlands.
Her research interests are free-space optical communications with a particular focus on channel modeling.
\end{IEEEbiographynophoto}

\begin{IEEEbiographynophoto}{Yunus Can Gültekin}
(S’10) was born in Izmir, Turkey. He received the B.Sc. and M.Sc. degrees in electrical and electronics engineering from Middle East Technical University, Ankara, Turkey, in 2013 and 2015, respectively.  He obtained his Ph.D. degree in 2020 from the Eindhoven University of Technology (TU/e), the Netherlands. During 2020-2023, he held a postdoctoral position at TU/e. Since January 2024, he has been a permanent research staff member at TU/e. His research focuses on understanding the requirements of future wireless/optical communications and quantum key distribution systems, and developing coded modulation and signal processing solutions for these systems using tools from information and communication theory. His research has received multiple awards, including Best Paper awards at the 2018 WIC/IEEE Symposium on Information Theory and Signal Processing in the Benelux and at the 2022 Optica Advanced Photonics Congress, a 2023 Quantum Delta NL exchange visit grant, and a 2024 International Excellence Fellowship from the Karlsruhe Institute of Technology.
\end{IEEEbiographynophoto}

\begin{IEEEbiographynophoto}{Kaiquan Wu} was born in Changsha, China, in 1994. He received his B.Sc. and M.Sc. in Communication Engineering from the College of Information Science and Engineering, Hunan University (HNU), China. He was in the state key laboratory of fiber communication during the Master period. He obtained his Ph.D. degree in 2023 from the Eindhoven University of Technology (TU/e), the Netherlands. His main research focuses on designing pragmatic and novel digital signal processing schemes, including error correction code, coded modulation, and channel modeling, for fiber-optic communication systems. 
\end{IEEEbiographynophoto}

\begin{IEEEbiographynophoto}{Cornelis Willem Korevaar} is a part-time university researcher at the Signal Processing Systems (SPS) Group. In addition to his work at the TU/e, he works at the Dutch Applied research institute TNO as a senior telecommunications system engineer in the optical satellite communications program. At TNO, he designs and develops optical and quantum satellite communication systems for ground-to-satellite and satellite-to-satellite communications. His main interests are in end-to-end system performance analysis, modem designs and signal processing. He is contributing to the standardization of space optical communication systems at the CCSDS and ESA. Wim received his PhD, cum laude, for his work on Hermite-based multi-user communication schemes, as an alternative to Fourier-based communication schemes. During his studies he founded two companies, which were acquired in 2007 and 2019, respectively

\end{IEEEbiographynophoto}

\begin{IEEEbiographynophoto}{Alex Alvarado}
 (S’06–M’11–SM’15) was born in Quellón, on the island of Chiloé, Chile. He received the Electronics Engineer’s degree (Ingeniero Civil Electrónico) and the M.Sc. degree (Magíster en Ciencias de la Ingeniería Electrónica) from Universidad Técnica Federico Santa María, Valparaíso, Chile, in 2003 and 2005, respectively, and the Licentiate of Engineering degree
(Teknologie Licentiatexamen) and Ph.D. degree from the Chalmers University of Technology, Gothenburg, Sweden, in 2008 and 2011, respectively.

During 2012–2014, he was a Marie Curie Intra-European Fellow with the University of Cambridge, Cambridge, U.K., and during 2011–2012, he was a Newton International Fellow with the same institution. He is currently a Senior Research Associate at the Optical Networks Group, University College London, London, U.K. His general research interests include digital communications, coding, and information theory. In 2008, he received the Merit Scholarship Program for Foreign Students, granted by the Ministere de l’Éducation, du Loisir et du Sports du Québec. He received the 2009 IEEE Information Theory Workshop Best Poster Award, the 2013 IEEE Communication Theory Workshop Best Poster Award, and the 2015 IEEE Transaction on Communications Exemplary
Reviewer Award
\end{IEEEbiographynophoto}

\vfill

\end{document}